\documentclass[lettersize,journal]{IEEEtran}
\usepackage{amsmath,amsfonts,amsthm}
\usepackage{algorithmic}
\usepackage{array}
\usepackage{bm}
\usepackage{bbm}
\usepackage[caption=false,font=normalsize,labelfont=sf,textfont=sf]{subfig}
\usepackage{textcomp}
\usepackage{stfloats}
\usepackage{braket}
\usepackage{threeparttable}
\usepackage{tablefootnote}
\usepackage{diagbox}
\usepackage{url}
\usepackage{hyperref}
\usepackage{romannum}
\usepackage{mathrsfs}%
\usepackage{verbatim}
\usepackage{multicol}
\usepackage{multirow}
\usepackage{caption}
\usepackage{bbding}
\usepackage[numbers,longnamesfirst]{natbib}
\usepackage{graphicx}
\newtheorem{theorem}{Theorem}%  meant for continuous numbers
\newtheorem{lemma}[theorem]{Lemma}
\newtheorem{corollary}[theorem]{Corollary}
\newtheorem{definition}{Definition}%

\renewcommand{\Re}{\operatorname{Re}}
\renewcommand{\Im}{\operatorname{Im}}
\DeclareMathOperator{\Tr}{Tr}

\def\BibTeX{{\rm B\kern-.05em{\sc i\kern-.025em b}\kern-.08em
    T\kern-.1667em\lower.7ex\hbox{E}\kern-.125emX}}
\usepackage{balance}
\begin{document}
\title{Universal 2-Local Symmetry-Preserving Quantum Neural Networks for Fermionic Systems}
\author{Ge Yan \textit{Member, IEEE}, Kaisen Pan, Ruocheng Wang, Mengfei Ran, Hongxu Chen, Junchi Yan \textit{Senior Member, IEEE }
\thanks{G. Yan is with Nanyang Technological University, Singapore. The other authors are with the School of Artificial Intelligence, Shanghai Jiao Tong University, Shanghai 200030, China (Correspondence author: Junchi Yan, Email: yanjunchi@sjtu.edu.cn).}}

\markboth{Universal 2-Local Symmetry-Preserving Quantum Neural Networks for Fermionic Systems}{}

% {IEEE Transactions on Pattern Analysis and Machine Intelligence,~Vol.~, No.~, September~2025}%

\maketitle

\begin{abstract}
Simulating quantum many-body systems represents a fundamental challenge where classical machine learning methods are severely bottlenecked by the exponential curse of dimensionality. Variational Quantum Algorithms (VQAs) offer a native paradigm to tackle this by optimizing parameterized unitary evolutions to find the ground states of problem Hamiltonians. However, the efficacy of these VQA is deeply hindered by the challenge of balancing the preservation of critical physical symmetries with the strict constraints of hardware implementability. In this work, we address this dilemma by proposing a hardware-efficient, symmetry-preserving ansatz fortified with complete theoretical guarantees for fermionic systems, termed the Hamming Weight Preserving (HWP) ansatz. We establish the necessary and sufficient conditions for 2-local HWP operators to achieve subspace universality, formally debunking the prevailing assumption that truncation-free simulation requires complex high-order interactions. Empirical validations corroborate our theoretical guarantees, showcasing the exact approximation of arbitrary unitary matrices within the HWP subspace. 
Crucially, we demonstrate the exceptional versatility of the proposed approach by deploying the exact same ansatz across distinct fermionic models, including diverse molecular electronic structures and the Fermi-Hubbard model. Our proposed HWP ansatz consistently suppresses ground-state energy errors below $1 \times 10^{-10}$ Ha, achieving a level of precision that surpasses the stringent threshold of chemical accuracy by multiple orders of magnitude. This work establishes a complete, theoretically fortified 2-local framework for symmetry-preserving computation, offering a highly universal and hardware-efficient building block for advancing quantum machine learning and fermionic many-body simulations.
\end{abstract}

\begin{IEEEkeywords}
Quantum Machine Learning, Quantum Many-body Simulation, Lie Algebra
\end{IEEEkeywords}

\section{Introduction}
\IEEEPARstart{S}{imulating} quantum many-body systems is the long-standing fundamental challenge across condensed matter physics~\cite{leblanc2015solutions}, quantum chemistry~\cite{helgaker2013molecular,cao2019quantum} and materials science~\cite{kotliar2006electronic}, as classical methods are bottlenecked by the curse of dimensionality~\cite{feynman2018simulating,kohn1999nobel}. Given the intrinsic quantum mechanical nature, quantum computers offer a native and mathematically robust paradigm for tackling these simulations~\cite{feynman2018simulating}. Within this domain, Quantum Machine Learning (QML)~\cite{biamonte2017quantum,cerezo2021variational}, particularly the Variational Quantum Eigensolver (VQE)~\cite{peruzzo2014variational}, has emerged as the premier algorithm, demonstrating significant potential for achieving practical utility on both near-term and Fault-tolerant devices~\cite{kandala2017hardware}. Operating as an ab initio approach~\cite{helgaker2013molecular}, VQE initializes from a trivial reference state and variationally optimizes a parameterized trial wavefunction (ansatz) to approximate the ground state, thereby extracting critical properties of the many-body system. Therefore, the efficacy of simulating many-body system with VQE is deeply constrained by the design of the variational ansatz~\cite{cerezo2021variational}.

There have been progressive improvements in variational ansatz design, including hardware-efficient~\cite{kandala2017hardware} and Hamiltonian variational types~\cite{romero1701strategies,stanisic2022observing}. Hardware-Efficient Ansätze (HEA) accommodate near-term devices but blindly explore the unconstrained Hilbert space. They suffer from degraded accuracy and severe barren plateaus since they ignore physical symmetries. Conversely, Hamiltonian-driven frameworks like Unitary Coupled Cluster (UCC)~\cite{anand2022sukin,bartlett1989alternative} inherently preserve symmetries but necessitate high-order, non-local interactions. Compiling these complex terms triggers an explosive growth in circuit depth, rendering them intractable even for anticipated fault-tolerant quantum computing  architectures.

To navigate the trade-off between hardware efficiency and physical meaning, a more robust approach is to leverage fundamental physical symmetries, which mathematically constrain the state evolution to dynamically isolated subspaces~\cite{gard2020efficient}. For fermionic systems, critical symmetries, such as particle-number conservation, naturally map to the preservation of the quantum-state Hamming weight under standard qubit encodings~\cite{Jordan1928,bravyi2002fermionic}. Existing symmetry-preserving circuit designs predominantly follow a bottom-up strategy, where heuristic operators are derived directly from the symmetry-preserving system Hamiltonian~\cite{wecker2015solving,jiang2018quantum,bacon2001encoded,terhal2002classical,farhi2014quantum,stanisic2022observing,gard2020efficient}. However, this fragmented approach lacks any rigorous theoretical guarantee regarding expressivity~\cite{zeier2011symmetry,larocca2023theory}. It remains unknown whether circuits composed of these ad-hoc operators can universally explore the subspace to reach the target ground state. Furthermore, it is conventionally believed that achieving truncation-free simulation of strongly correlated many-body systems strictly necessitates non-local, high-order interaction terms \cite{romero1701strategies,anand2022sukin}. However, compiling such complex operators imposes prohibitive circuit depth overheads, constituting a major bottleneck for the practical deployment of VQEs \cite{tilly2022variational}. Therefore, discovering a hardware-efficient, symmetry-preserving ansatz fortified with strict theoretical guarantees remains a central and highly practical challenge in QML.

In this work, we address this critical gap by presenting a novel top-down framework for constructing mathematically interpretable Hamming Weight Preserving (HWP) ansätze. Diverging from prior bottom-up methods that derive scattered gates from specific Hamiltonians, our paradigm systematically delineates the complete mathematical structure of the HWP space. From this universal foundation, we rigorously deduce the necessary and sufficient conditions required for 2-local operators to achieve complete controllability over the constrained subspace. This top-down derivation renders our ansatz entirely independent of any specific system Hamiltonian, effectively providing a universal, truncation-free building block capable of simulating diverse fermionic many-body systems without resorting to complex high-order terms.

Specifically, we first derive the necessary and sufficient conditions for 2-local HWP operators to achieve universality in the HWP subspace by analyzing the Lie Algebra dimension, ensuring maximum expressivity. We then establish the trainability of the ansatz by deriving exact gradient variance expressions \citep{mcclean2018barren, wang2021noise}. This theoretical analysis rigorously quantifies the optimization landscape, demonstrating that the gradient vanishing is strictly bounded by the physical subspace dimension rather than the full unconstrained Hilbert space. The theoretical findings are validated by successfully approximating arbitrary unitary matrices within the HWP subspace \citep{nielsen2002quantum, chong2017programming, khatri2019quantum}, demonstrating the ansatz's ability to efficiently capture target quantum states or transformations. To further showcase the practical utility of the proposed ansatz, we evaluate on representative Fermionic systems~\citep{kandala2017hardware, guo2024experimental, hubbard1963electron}, including solving electronic structures for molecules and the Fermi-Hubbard model. The proposed ansatz achieves remarkable accuracy on both models and suppresses the Hamiltonian variational ansatz in terms of both efficiency and accuracy.

\section{Preliminary}
\subsection{Hamming Distance and Hamming Weight}
Considering two binary vectors $\mathbf{a}$ and $\mathbf{b}$ with $\mathbf{a}, \mathbf{b}\in \{0,1\}^{N}$, where $N$ is the dimension of the vectors. The Hamming distance of these two vectors is as follows:
\begin{definition}\label{def:hammingdistance}
\textbf{Hamming distance:} The Hamming distance $\mathcal{D}$ of two binary vectors $\mathbf{a}$ and $\mathbf{b}$ is:
\begin{equation}
\begin{aligned}
    \mathbf{c}=\mathbf{a} \oplus \mathbf{b},\\
    \mathcal{D}(\mathbf{a},\mathbf{b})=\sum_{i=1}^N \mathbf{c}_i,
\end{aligned}
\end{equation}
where $\oplus$ stands for exclusive OR operator.
\end{definition}
With the definition of Hamming distance, we can further define the Hamming weight of a given binary vector $\mathbf{a}$.
\begin{definition}\label{def:hammingweight}
\textbf{Hamming weight:} Let $\mathbf{0}=0^N$. The Hamming weight $HW$ of binary vector $\mathbf{a}$ is:
\begin{equation}
    HW(\mathbf{a})=\mathcal{D}(\mathbf{a},\mathbf{0}),
\end{equation}
\end{definition}

\subsection{Lie Algebra and Controllability of Quantum Systems}
We first introduce a well-defined mathematical tool to derive the theoretical findings, which is Dynamical Lie Algebra (DLA)~\citep{zeier2011symmetry,d2021introduction}. Lie algebraic techniques have been widely used in discussing the controllability of quantum systems due to the unitary transformation nature of the quantum circuits. For a more intuitive understanding, the Dynamical Lie Group $\mathbb{G}$ comprises unitary matrices, and the corresponding DLA $\mathfrak{g}$ comprises anti-Hermitian matrices. For a $L$-gate parameterized quantum circuit, we can denote the ansatz as a unitary transformation as
\begin{equation}\label{eq:PQC}
    \mathbf{U}(\bm{\theta})=\prod_{l=1}^L \mathbf{U}_l({\theta}_l)=\prod_{l=1}^L e^{\text{i}\theta_{l}\mathbf{H}_l},
\end{equation}
where $\mathbf{U}_l=e^{\text{i}\mathbf{H}_l\theta_l}$, $\text{i}\mathbf{H}_l$ are anti-Hermitian matrices, and $\bm\theta=\{\theta_1,\theta_2\cdots\theta_L\}$ are the parameters. We take these distinct Hamiltonians in the circuit as a set of generators $\mathcal{G}=\{\text{i}\mathbf{H}_p\}_{p=1}^P$, where $|\mathcal{G}|=P$. The DLA can be defined as
\begin{definition}\label{def:DLA}
\textbf{Dynamical Lie Algebra (DLA):} Consider the set of generators $\mathcal{G}=\{\text{i}\mathbf{H}_p\}_{p=1}^P$, the DLA $\mathfrak{g}$ is defined as:
\begin{equation}
    \mathfrak{g}=span\langle \rm{i}\mathbf{H}_1, \rm{i}\mathbf{H}_2,\cdots,\rm{i}\mathbf{H}_P\rangle_{Lie},
\end{equation}
where $\langle\cdot\rangle_{Lie}$ denotes the Lie closure. 
\end{definition}
The DLA $\mathfrak{g}$ is calculated by repeatedly taking the commutator of the elements in the set of generators. The commutator of matrices $\mathbf{A}$ and $\mathbf{B}$ can be defined as $[\mathbf{A},\mathbf{B}]=\mathbf{AB}-\mathbf{BA}$. 
The reachable unitary matrices of the parameterized quantum ansatz with arbitrary parameters $\bm\theta$ can be denoted as $\{\mathbf{U}(\bm{\theta})\}_{\bm{\theta}}$. We introduce the following lemma
\begin{lemma}\label{lemma:controllability}
    ~\cite{ramakrishna1995controllability}  A quantum system $\hat{\mathbf{H}}$ is completely controllable if $\{\mathbf{U}(\bm{\theta})\}_{\bm{\theta}}= \mathbb{G}=\mathcal{SU}(N)$.
\end{lemma}
Complete controllability indicates that any desired quantum state evolution can be achieved using the ansatz. In other words, the system's dynamics can be fully manipulated to reach any state in the given space from any initial state. We further borrow an important conclusion from~\cite{ramakrishna1995controllability}, which provides a simple way to verify the controllability of a quantum system, by computing the dimension of the algebra.
\begin{lemma}\label{lemma:control}
~\cite{ramakrishna1995controllability} A necessary and sufficient condition for complete controllability of a quantum system $\hat{\mathbf{H}}$ is that the dimension of the DLA $\mathfrak{g}$ is $N^2$. 
\end{lemma}
This lemma is then utilized to derive the conditions for universal HWP operators in the theoretical results.

\begin{figure*}
    \centering
    \includegraphics[width=\linewidth]{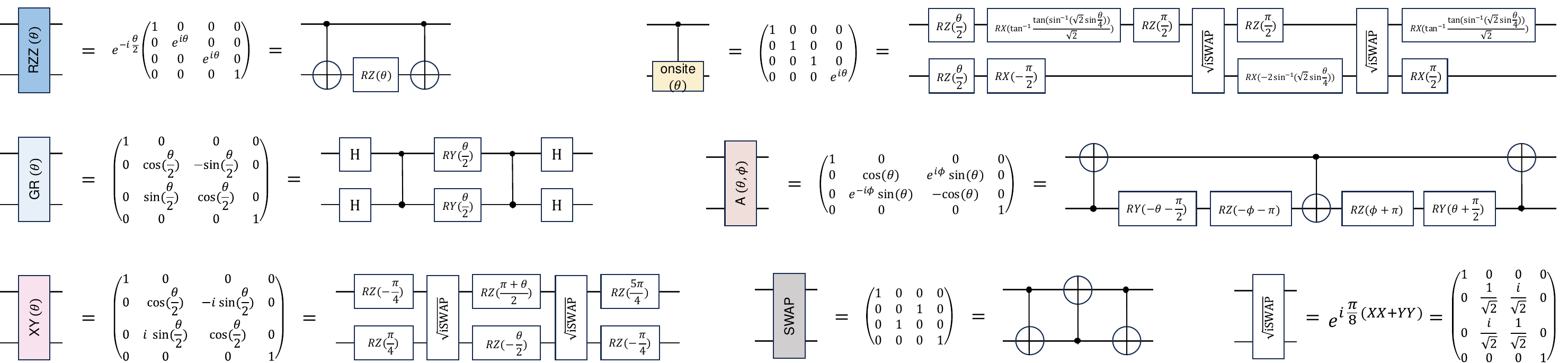}
    \caption{\textbf{Two-qubit HWP gates in related works with their unitary matrices and possible decomposition.} The two-qubit elementary gates utilized in the decomposition include CX, CZ, and $\sqrt{\text{iSWAP}}$ (unitary matrix in the bottom right). We have RZZ~\cite{farhi2014quantum}, Givens Rotations~\cite{wecker2015solving,jiang2018quantum}, XY-interaction~\cite{bacon2001encoded,terhal2002classical}, onsite~\cite{stanisic2022observing}, $A$ gate~\cite{gard2020efficient}, and non-parameterized SWAP gate. }
    \label{fig:HWPgates}
\end{figure*}
\subsection{Hamming Weight Preserving Operators}
Based on the definition of Hamming weight and Hamming distance, we can thus define a type of quantum circuit, namely Hamming Weight Preserving quantum circuit, where the number of non-zero elements in the quantum state remains invariant throughout its evolution. This type of circuit naturally enforces symmetry within the model by treating it as a hard constraint on the circuit. To ensure that the quantum circuit preserves the Hamming weight of the states, the Hamiltonian of two-qubit HWP gates takes the following form:
\begin{equation}\label{eq:methodHW}\small
    \mathbf{H}_{HW}= \left(\begin{array}{cccc}
        0 & 0 & 0 & 0\\
        0 & a & b & 0\\
        0 & \bar{b} & c & 0\\
        0 & 0 & 0 & 0
    \end{array}\right),
\end{equation}
where $a$ and $c$ are real values and $\bar{b}$ denotes the complex conjugate of $b$. 
Two-qubit HWP gates with this structure act exclusively on the basis states $\ket{01}$ and $\ket{10}$, thereby maintaining the Hamming weight of the quantum state throughout the evolution. This characteristic not only enforces particle number symmetry but also provides resilience against bit-flip errors, as any such error would alter the Hamming weight, making it easily detectable with a simple parity check. Furthermore, incorporating symmetry verification as an additional layer can enhance the robustness and performance of the HWP ansatz on near-term devices.
In Figure~\ref{fig:HWPgates}, we provide several two-qubit gates that have been identified as HWP gates in previous works~\cite{wecker2015solving,jiang2018quantum,bacon2001encoded,terhal2002classical,farhi2014quantum,stanisic2022observing,gard2020efficient}, fitting the form in Equation~\ref{eq:methodHW}. However, none of these gates have been rigorously analyzed in terms of their expressivity, and their performance has been evaluated only on specific tasks with limited ansatz layers, leaving them underparameterized. 

Quantum computing is inherently susceptible to noise, which poses significant challenges for reliable computation. Extensive research has been devoted to quantum error mitigation (QEM)~\cite{huang2023near,cai2023quantum} and quantum error correction~\cite{fowler2012surface,kovalev2013fault,breuckmann2021quantum,bausch2024learning}. While these efforts are essential for universal quantum computing, some works have also focused on developing quantum ansatz that are self-protected or resilient to specific types of errors~\cite{viola2001experimental,pyshkin2022self}. Such approaches are particularly valuable for the near-term advancement of quantum processors. The HWP ansatz explored in this paper is inherently resistant to bit-flip errors and is well-suited for integration with QEM techniques, such as symmetry verification~\cite{yan2024rethinking}. This inherent compatibility enhances its practical utility, simplifying error mitigation and correction in near-term quantum devices.

\section{Related Works}
We first review existing QML~\cite{biamonte2017quantum} algorithms for simulating fermionic systems. A foundational approach is the Hamiltonian variational ansatz~\cite{wecker2015solving}, which is designed to construct a physically rigorous trial wavefunction. Since the problem Hamiltonian, which encodes the total energy of the many-body system, is Hermitian, its exponentiation yields unitary evolution operators $U=e^{-iHt}$ that can be seamlessly implemented on quantum circuits, with time $t$ serving as the variational parameter. Consequently, the variational layers are constructed by alternately applying the constituent terms of the problem Hamiltonian to form the complete ansatz. However, the implementation complexity of this approach is strictly bottlenecked by the complexity of the Hamiltonian terms themselves~\cite{anand2022sukin,guo2024experimental}. For strongly correlated systems, the problem Hamiltonian frequently contains high-order, non-local interaction terms that are notoriously difficult to decompose into the basic built-in gates available on quantum hardware~\cite{ryabinkin2018qubit}. To demonstrate proof-of-principle advantages on near-term quantum processors, the Hardware-Efficient Ansatz (HEA) was proposed~\cite{kandala2017hardware}. Constructed exclusively from native, parameterized single-qubit rotations (e.g., $R_Y, R_Z$) and entangling CNOT gates, the HEA prioritizes executability. Yet, because it does not enforce physical symmetries of any kind, the HEA blindly explores the entire exponentially large Hilbert space. This unrestricted exploration not only risks converging to unphysical final states but also exacerbates severe barren plateaus during the training phase~\cite{cerezo2021cost}.

To mitigate these structural design challenges, various machine learning techniques have been introduced to dynamically optimize the ansatz architecture. A prominent paradigm is Quantum Architecture Search (QAS)~\cite{yan2024quantum}, which draws inspiration from classical neural architecture search by simultaneously optimizing both the circuit topology and its parameterized gates using a unified loss function. Nevertheless, most QAS methods adhere to the HEA design philosophy, utilizing native hardware gates as the candidate pool, to minimize CNOT counts for near-term implementation~\cite{du2022quantum,wu2023quantumdarts}. As a result, they inherently fail to generate symmetry-preserving ansätze. Conversely, adaptive algorithms such as ADAPT-VQE~\cite{grimsley2019adaptive,tang2021qubit} populate their candidate gate pool with Hamiltonian variational operators. By dynamically adjusting the selection and application order of these operators, ADAPT-VQE can reduce the total number of required gates while strictly preserving system symmetries. However, the primary obstacle impeding Hamiltonian-driven ansätze is not merely the quantity of stacked operators, but the intrinsic implementational complexity of the high-order operators themselves. Thus, an optimal balance between physical validity and hardware implementability remains elusive.

Parallel to architectural innovations, continuous efforts have been directed toward rigorously understanding the training dynamics of QML algorithms. The fundamental connection between quantum optimal control and Lie algebra was formally established in~\cite{ramakrishna1995controllability}, laying the groundwork for analyzing quantum system controllability. Building upon this, \cite{du2020expressive} elaborated on the expressivity of quantum ansätze, while \cite{larocca2023theory} investigated the theory of overparameterization by analyzing the Quantum Fisher Information Matrix and DLA. Concurrently, the geometry of the optimization landscape, specifically the barren plateau phenomenon, has remained a major bottleneck and a subject of intense scrutiny~\cite{larocca2025barren}. To be both practically viable and mathematically sound, a robust quantum ansatz must be rigorously evaluated against all these theoretical criteria: subspace expressivity, trainability, and hardware-efficient gate complexity.

\textit{Remarks on Concurrent and Preliminary Works:} It is worth noting a concurrent study~\cite{monbroussou2025trainability} that also explores the DLA of Hamming-weight preserving circuits. However, their framework fundamentally diverges from ours in both application and algebraic completeness. They primarily deploy the HWP ansatz as a quantum data loader for classical data embedding. This application is highly counterintuitive, as classical data distributions rarely possess intrinsic, strict Hamming-weight symmetries, rendering the subspace constraint physically unmotivated. Furthermore, their analysis heavily relies on the Givens Rotation (GR) gate, which strictly confines the system's evolution to real-valued orthogonal sub-manifolds. This structural deficit fundamentally prevents their ansatz from spanning the complete special unitary group $SU(d_k)$ required for absolute many-body simulation. Meanwhile, in our preliminary conference version~\cite{yan2024HWP}, we initially identified the alignment between the HWP ansatz and constrained VQAs. This journal manuscript significantly generalizes and formally completes that framework. Diverging from heuristic observations, we now provide the rigorous top-down algebraic derivation of the 2-local BS gate, formal proofs for necessary and sufficient conditions under restrictive Nearest-Neighbor topologies, exact gradient variance bounds characterizing trainability, and extensive scalability validations on strongly correlated Fermi-Hubbard systems.

\section{Methodology}
\subsection{Theoretical Results for Universality}
We begin by establishing the necessary and sufficient conditions for an HWP ansatz to be universal within the HWP subspace. By leveraging quantum optimal control theory~\cite{ramakrishna1995controllability}, we identify these conditions through the analysis of the Dynamical Lie Algebra (DLA) generated by the HWP ansatz.
The general form of a 2-local HWP gate in Equation~\ref{eq:methodHW} can be decomposed with the following four basis matrices:
\begin{equation}\begin{aligned}
    \mathbf{R}_{ij}&= \begin{pmatrix}\begin{smallmatrix}
        0 & 0 & 0 & 0\\
        0 & 0 & 1 & 0\\
        0 & 1 & 0 & 0\\
        0 & 0 & 0 & 0\end{smallmatrix}
    \end{pmatrix},\ 
    \mathbf{J}_{ij}= \begin{pmatrix}\begin{smallmatrix}
        0 & 0 & 0 & 0\\
        0 & 0 & \text{i} & 0\\
        0 & -\text{i} & 0 & 0\\
        0 & 0 & 0 & 0\end{smallmatrix}
    \end{pmatrix},\\
    \mathbf{E}_{ij}&=\begin{pmatrix}\begin{smallmatrix}
        0 & 0 & 0 & 0\\
        0 & 1 & 0 & 0\\
        0 & 0 & 1 & 0\\
        0 & 0 & 0 & 0\end{smallmatrix}
    \end{pmatrix},\ 
    \mathbf{S}_{ij}= \begin{pmatrix}\begin{smallmatrix}
        0 & 0 & 0 & 0\\
        0 & 1 & 0 & 0\\
        0 & 0 & -1 & 0\\
        0 & 0 & 0 & 0\end{smallmatrix}
    \end{pmatrix},
    \end{aligned}
\end{equation}
\begin{equation}\label{eq:HWmat}
    \mathbf{H}_{HWij}=\frac{a+c}{2}\mathbf{E}_{ij}+\frac{a-c}{2}\mathbf{S}_{ij}+\frac{b+\bar{b}}{2}\mathbf{R}_{ij}+\frac{b-\bar{b}}{2\text{i}}\mathbf{J}_{ij}.
\end{equation}
To simplify the decomposition in Equation~\ref{eq:HWmat}, we denote $r = \Re b$, $j = \Im b$, $e = \frac{a+c}{2}$, and $s = \frac{a-c}{2}$, where $r$ and $j$ are the real and imaginary parts of $b$, respectively. 
\begin{equation}\label{eq:HWmat_simple}
    \mathbf{H}_{HWij}=e\mathbf{E}_{ij}+s\mathbf{S}_{ij}+r\mathbf{R}_{ij}+j\mathbf{J}_{ij}.
\end{equation}
For an HWP subspace of $n$ qubits with an Hamming weight  of $k$, the dimension is given by $d_k = \tbinom{n}{k}$. %The HWP ansatz is constructed using two-qubit HWP operators, and the number of initial generators for the DLA depends on the physical qubit connectivity. 
Since the number of initial generators for the DLA depends on the physical qubit connectivity, we begin by analyzing an ideal case where two-qubit HWP gates are applied between any pair of qubits. The following Theorem outlines the conditions under which a two-qubit HWP gate is universal when full connectivity (FC) is available. These conditions provide a rigorous framework for understanding the expressiveness of the HWP ansatz.

\begin{table}[t]  \centering
\caption{\textbf{The commutation results of basis matrices.} \textbf{Top:} commutator on same qubits; \textbf{Bottom:} commutator on different qubits. Coefficients ommited.}\label{tab:comm_basis}
\begin{minipage}{0.55\linewidth}\renewcommand{\arraystretch}{1.2}
    \resizebox{\textwidth}{!}{\begin{tabular}{c|c|c|c|c} 
        \hline
         & $\mathbf R_{ij}$ & $\mathbf J_{ij}$ & $\mathbf E_{ij}$ & $\mathbf S_{ij}$ \\
        \hline
         $\mathbf R_{i j}$&$0$     &$\mathbf S_{ij}$&$0$     &$\mathbf J_{ij}$ \\
         $\mathbf J_{ij}$ & &$0$     &$0$     &$\mathbf R_{ij}$ \\
         $\mathbf E_{ij}$ &  &    &$0$     &$0$      \\
         $\mathbf S_{ij}$ && &    &$0$      \\
        \hline
    \end{tabular}}
\end{minipage}\par \vspace{0.1cm}
\begin{minipage}{0.8\linewidth}\renewcommand{\arraystretch}{1.2}
    \resizebox{\textwidth}{!}{\begin{tabular}{c|c|c|c|c} 
        \hline
         &$\mathbf R_{ij}$                   &$\mathbf J_{ij}$                   &$\mathbf E_{ij}$                   &$\mathbf S_{ij}$ \\
         \hline
         $\mathbf R_{jk}$&$\mathbf J_{ik}\otimes \sigma^z_j$&$\mathbf R_{ik}\otimes \sigma^z_j$&$\mathbf J_{jk}\otimes \sigma^z_i$&$\mathbf J_{jk}$ \\
         $\mathbf J_{jk}$&  &$\mathbf J_{ik}\otimes \sigma^z_j$&$\mathbf R_{jk}\otimes \sigma^z_i$&$\mathbf R_{jk}$ \\
         $\mathbf E_{jk}$& & &$0$                        &$0$      \\
         $\mathbf S_{jk}$&                 &           &             &$0$      \\
         \hline
    \end{tabular}}
\end{minipage}
\end{table}
\begin{table}[t]  \centering
\caption{\textbf{The commutation results of basis matrices tensored with Pauli-Z.} \textbf{Top:} only one with Pauli-Z; \textbf{Bottom:} both tensored with Pauli-Z. Coefficients ommited.}\label{tab:comm_z}
\begin{minipage}{0.81\linewidth}\centering\renewcommand{\arraystretch}{1.3}
    \resizebox{\textwidth}{!}{\begin{tabular}{c|c|c|c|c} 
        \hline
        &$\mathbf R_{ij}\otimes\sigma^z_k$&$\mathbf J_{ij}\otimes\sigma^z_k$&$ \mathbf E_{ij}\otimes\sigma^z_k$&$\mathbf S_{ij}\otimes\sigma^z_k$  \\
        \hline
        $\mathbf R_{jk}$         &         $\mathbf J_{ik}$         &         $\mathbf R_{ik}$     &  $\mathbf J_{jk}$  &$\mathbf J_{jk}\otimes\sigma^z_i$  \\
        $\mathbf J_{jk}$         &   $\mathbf{R}_{ik}$   &         $\mathbf J_{ik}$    &   $\mathbf R_{jk}$   &$\mathbf R_{jk}\otimes \sigma^z_i$ \\
        $\mathbf E_{jk}$         &   $\mathbf J_{ij}$       &   $\mathbf R_{ij}$   &    $0$ &        $0$             \\
        $\mathbf S_{jk}$         & $\mathbf J_{ij}\otimes\sigma^z_k$ & $\mathbf R_{ij}\otimes\sigma^z_k$  & $0$ &           $0$    \\
        \hline
    \end{tabular}}
\end{minipage}\par \vspace{0.1cm}
\begin{minipage}{0.9\linewidth}\renewcommand{\arraystretch}{1.3}
    \resizebox{\textwidth}{!}{\begin{tabular}{c|c|c|c|c}
        \hline
        &$\mathbf R_{ij}\otimes\sigma^z_k$&$\mathbf J_{ij}\otimes\sigma^z_k$&$ \mathbf E_{ij}\otimes\sigma^z_k$&$\mathbf S_{ij}\otimes\sigma^z_k$  \\  
        \hline
        $\mathbf R_{jk}\otimes\sigma^z_i$&$\mathbf J_{ik}\otimes\sigma^z_j$&$\mathbf R_{ik}\otimes\sigma^z_j$&    $\mathbf J_{jk}\otimes\sigma^z_i$ &     $\mathbf J_{jk}$    \\ 
        $\mathbf J_{jk}\otimes\sigma^z_i$&&$\mathbf J_{ik}\otimes\sigma^z_j$&  $\mathbf R_{jk}\otimes\sigma^z_i$  &     $\mathbf R_{jk}$    \\
        $\mathbf S_{jk}\otimes\sigma^z_i$&&&    $0$ &     $0$    \\
        $\mathbf E_{jk}\otimes\sigma^z_i$&&&&     $0$    \\
        \hline
    \end{tabular}}
\end{minipage}
\end{table}

\begin{theorem}\label{thm:main}
    For any $n$ and $k$, a two-qubit HWP gate is universal with FC if and only if the coefficients satisfy one of the following two conditions:
    \begin{equation}
            \text{(1) }e\neq0,j\neq0;\text{\quad\quad (2) }e\neq0,r\neq0,s\neq0.
            \label{eq:theo}
    \end{equation}
\end{theorem}
\begin{proof}
The most critical point to understand the results is by analyzing the commutators of these four basis matrices. We begin by analyzing the commutators for two qubits in the same set ($ij$ and $ij$) and for different qubits ($ij$ and $jk$). It is important to note that basis matrices acting on entirely distinct qubits, such as $ij$ and $kl$, do not commute, and therefore this case is excluded from our analysis.
As shown in Table~\ref{tab:comm_basis}, the commutators between basis matrices on $ij$ and $jk$ involve the Pauli-Z operator. Consequently, we extend the analysis to include commutators between basis matrices and those tensored with the Pauli-Z operator. The results in Table~\ref{tab:comm_z} demonstrate that these commutators form a closed algebra, generating no new elements.

There are fifteen distinct configurations of the HWP Hamiltonian, obtained by combining the four basis matrices (excluding the trivial all-zero case). Utilizing the commutation relations derived earlier, we compute the final commutation outcomes for these configurations and determine the corresponding DLA dimensions. We now provide the derivation of the eight distinct DLA types that cover all fifteen configurations. The directed arrows in figure~\ref{fig:eighttype} indicate generative relations of the elements. 
\begin{figure}
    \centering
    \includegraphics[width=\linewidth]{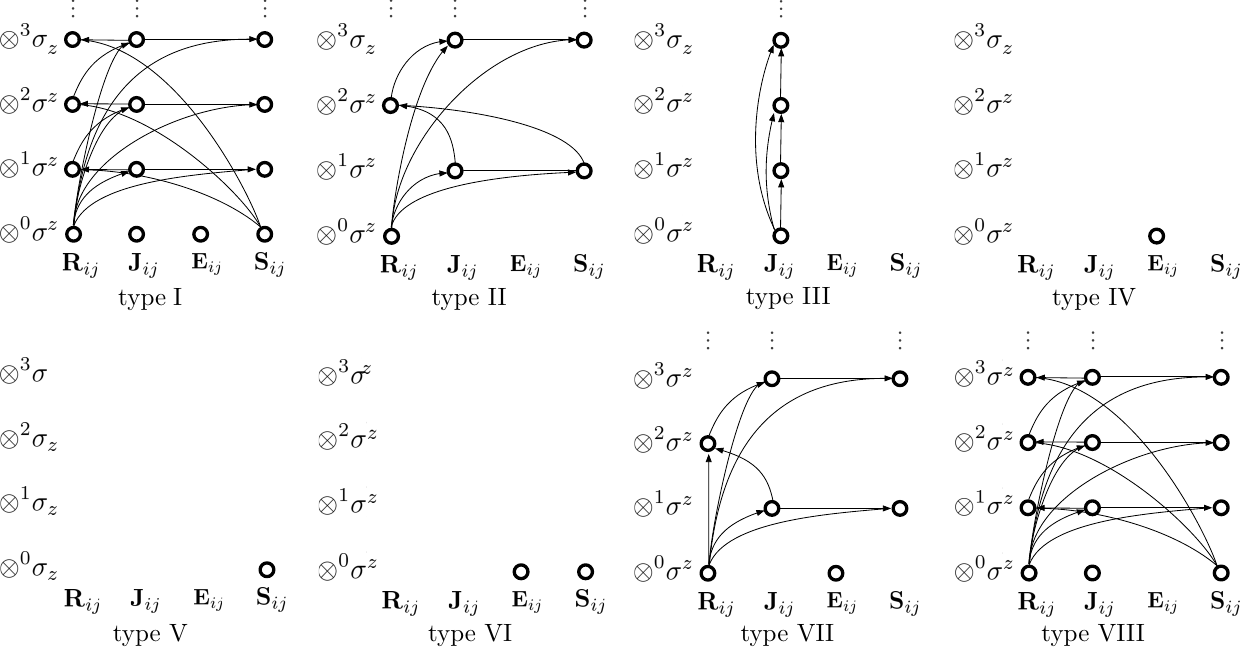}
    \caption{\textbf{Generative relationships of the basis matrices.} Directed arrows indicate the generative pathways via commutators. Although successive commutations progressively tensor additional Pauli-Z ($\sigma_z$) operators, they strictly do not yield any new type of basis matrix, thereby explicitly demonstrating the closure of the generated algebra.}
    \label{fig:eighttype}
\end{figure}
Generators in the DLA type \Romannum{1} lead to the following set of generative relationships:
\begin{equation}
    \begin{aligned}
        [\mathbf R_{ij},\mathbf R_{jk}\otimes^l \sigma^z_{m_l\neq i,j,k}]&=\text i\mathbf J_{ki}\otimes\sigma^z_j\otimes^l \sigma^z_{m_l\neq i,j,k}\\&=\text i\mathbf J_{ki}\otimes^{l+1}\sigma^z_{m_l\neq i,k}, \\
        [\mathbf J_{ki}\otimes^{l+1}\sigma^z_{m_l\neq i,k},\mathbf R_{ki}]&=2\text i\mathbf S_{ki}\otimes^{l+1}\sigma^z_{m_l\neq i,k},\\
        [\mathbf S_{ki},\mathbf J_{ki}\otimes^{l+1}\sigma^z_{m_l\neq i,k}]&=2\text i\mathbf R_{ki}\otimes^{l+1}\sigma^z_{m_l\neq i,k}.
    \end{aligned}
\end{equation}
Generators in the DLA type \Romannum{2} lead to the following set of generative relationships:
\begin{equation}
    \begin{aligned}
        [\mathbf R_{ij},\mathbf R_{jk}\otimes^{2l} \sigma^z_{m_l\neq i,j,k}]&=\text i\mathbf J_{ki}\otimes\sigma^z_j\otimes^{2l} \sigma^z_{m_l\neq i,j,k}\\&=\text i\mathbf J_{ki}\otimes^{2l+1}\sigma^z_{m_l\neq i,k},\\
        [\mathbf J_{ki}\otimes^{2l+1}\sigma^z_{m_l\neq i,k},\mathbf R_{ik}]&=2\text i\mathbf S_{ki}\otimes^{2l+1}\sigma^z_{m_l\neq i,k},\\
        [\mathbf S_{ki}\otimes\sigma^z_{j},\mathbf J_{ki}\otimes^{2l+1}\sigma^z_{m_l\neq i,k}]&=2\text i\mathbf R_{ki}\otimes\sigma^z_{j}\otimes^{2l+1}\sigma^z_{m_l\neq i,j,k}\\&=2\text iR_{ki}\otimes^{2l+2}\sigma^z_{m_l\neq i,j,k}.
    \end{aligned}
\end{equation}
Generators in the DLA type \Romannum{3} lead to the following set of generative relationships:
\begin{equation}
\begin{aligned}
    [\mathbf J_{ij},\mathbf J_{jk}\otimes^{l} \sigma^z_{m_l\neq i,j,k}]&=\text i\mathbf J_{ik}\otimes\sigma^z_j\otimes^{l} \sigma^z_{m_l\neq i,j,k}\\&=\text i\mathbf J_{ik}\otimes^{2l+1}\sigma^z_{m_l\neq i,k}
    \end{aligned}
\end{equation}
Generators in the DLA type \Romannum{4}, \Romannum{5} and \Romannum{6} are simple configurations with no further generation.\\
Generators in the DLA type \Romannum{7} lead to the following set of generative relationships:
\begin{equation}
    \begin{aligned}
        [\mathbf R_{ij},\mathbf R_{jk}\otimes^{2l} \sigma^z_{m_l\neq i,j,k}]&=\text i\mathbf J_{ki}\otimes\sigma^z_j\otimes^{2l} \sigma^z_{m_l\neq i,j,k}\\&=\text i\mathbf J_{ki}\otimes^{2l+1}\sigma^z_{m_l\neq i,k},\\
        [\mathbf J_{ki}\otimes^{2l+1}\sigma^z_{m_l\neq i,k},\mathbf R_{ik}]&=2\text i\mathbf S_{ki}\otimes^{2l+1}\sigma^z_{m_l\neq i,k},\\
        [\mathbf S_{ki}\otimes\sigma^z_{j},\mathbf J_{ki}\otimes^{2l+1}\sigma^z_{m_l\neq i,k}]&=2\text i\mathbf R_{ki}\otimes\sigma^z_{j}\otimes^{2l+1}\sigma^z_{m_l\neq i,j,k}\\&=2\text iR_{ki}\otimes^{2l+2}\sigma^z_{m_l\neq i,j,k}
    \end{aligned}
\end{equation}
Generators in the DLA type \Romannum{8} lead to the following set of generative relationships:
\begin{equation}
    \begin{aligned}
        [\mathbf R_{ij},\mathbf R_{jk}\otimes^l \sigma^z_{m_l\neq i,j,k}]&=\text i\mathbf J_{ki}\otimes\sigma^z_j\otimes^l \sigma^z_{m_l\neq i,j,k}\\&=\text i\mathbf J_{ki}\otimes^{l+1}\sigma^z_{m_l\neq i,k},\\
        [\mathbf J_{ki}\otimes^{l+1}\sigma^z_{m_l\neq i,k},\mathbf R_{ki}]&=2\text i\mathbf S_{ki}\otimes^{l+1}\sigma^z_{m_l\neq i,k},\\
        [\mathbf S_{ki},\mathbf J_{ki}\otimes^{l+1}\sigma^z_{m_l\neq i,k}]&=2\text i\mathbf R_{ki}\otimes^{l+1}\sigma^z_{m_l\neq i,k}
    \end{aligned}
\end{equation}
The dimension of these eight DLA types is demonstrated in Table~\ref{tab:DLAtype}. Notice that $k\leq n/2$ in the Table since the dimension of DLA $\{n,k\}$ is the same as $\{n,n-k\}$. We then derive all fifteen configurations (excluding the all-zero one) and illustrate how they can be regulated to these eight DLA types. The key observation is that the generator $\bm{E}_{ij}$ can not be generated by the combination of other three basis matrices. \\
1) $r,\ j\neq0, e=s=0$: From Table~\ref{tab:comm_basis}, $[\mathbf R_{ij},\mathbf J_{ij}]=\mathbf S_{ij}$. Leading to type \Romannum{7}.\\
2) $r,\ s\neq0, e=j=0$: From Table~\ref{tab:comm_basis}, $[\mathbf R_{ij},\mathbf S_{ij}]=\mathbf J_{ij}$. Leading to type \Romannum{7}.\\
3) $s,\ j\neq0, e=r=0$: From Table~\ref{tab:comm_basis}, $[\mathbf S_{ij},\mathbf J_{ij}]=\mathbf R_{ij}$. Leading to type \Romannum{7}.\\
4) $r,\ j,\ e\neq0, s=0$: Complimenting 1) with $\mathbf E_{ij}$. Leading to type \Romannum{1}.\\
5) $r,\ s,\ e\neq0, j=0$: Complimenting 2) with $\mathbf E_{ij}$. Leading to type \Romannum{1}.\\
6) $s,\ j,\ e\neq0, r=0$: Complimenting 3) with $\mathbf E_{ij}$. Leading to type \Romannum{1}.\\
7) $e,\ j\neq0,\ s=r=0$: From Table~\ref{tab:comm_basis} and Table~\ref{tab:comm_z}, 
\begin{equation}
    \begin{split}
        [\mathbf E_{ij},\mathbf J_{jk}]&=\text{i}(\mathbf R_{jk}\otimes\sigma^z_i)\\
        [\mathbf R_{ij}\otimes\sigma^z_k,\mathbf J_{jk}]&=\text{i}\mathbf R_{ik}\\
        [\mathbf R_{ik},\mathbf J_{ik}]&=-2\text{i}\mathbf S_{ik}.
    \end{split}
\end{equation}
This then lead to all four basis matrices as generators, which is type \Romannum{1}.

To sum up, we can conclude the classification of the fifteen configurations to the eight DLA types as shown in Table~\ref{tab:DLAcfg}. Of all 15 configurations, only five yield full dimensionality in the FC configuration.

% We have the following critical observations regarding the Hamiltonian in Equation~\ref{eq:HWmat}.
% \begin{equation}
% [\bm{H}_{HWij},\bm{H}_{HWji}]=4 \text i r( I\mathbf S_{ij} -  s   \mathbf J_{ij})
% \label{eq:detsj1}
% \end{equation}
% \begin{equation}
% \bm{H}_{HWij}+\bm{H}_{HWji}=2e\mathbf E_{ij}+2r\mathbf R_{ij}
% \end{equation}
% \begin{equation}
% \bm{H}_{HWij}-\bm{H}_{HWji}=2j\mathbf J_{ij}+2s\mathbf S_{ij}
% \label{eq:detsj2}
% \end{equation}

\begin{table*}[tb]
\begin{center}
\caption{Eight DLA types and the corresponding dimension. Checkmark denotes basis matrix in the generator set $\mathcal{G}$.}
\label{tab:DLAtype}
\resizebox{\textwidth}{!}{
\begin{tabular}{c|c|c|c|c|c|c|c|c} 
\hline
&\Romannum{1}&\Romannum{2}&\Romannum{3}&\Romannum{4}&\Romannum{5}&\Romannum{6}&\Romannum{7}&\Romannum{8}\\
\hline
$\mathbf R_{ij}$&\Checkmark&\Checkmark&&& & &\Checkmark&\Checkmark\\
$\mathbf J_{ij}$&\Checkmark&&\Checkmark& & &&&\Checkmark\\
$\mathbf E_{ij}$&\Checkmark& &&\Checkmark& &\Checkmark&\Checkmark&\\
$\mathbf S_{ij}$&\Checkmark & &&&\Checkmark& \Checkmark&&\Checkmark\\
\hline
dim
&$d_k^2$
&\tiny{$\left\{\begin{aligned}&d_k^2-1&\text{if } k< n/2\\&\frac{d_k^2}2-2&\text{if }k=n/2\end{aligned}\right.$}
&$\frac{d_k(d_k-1)}2$ 
&\tiny{$\left\{\begin{aligned}&n&\text{if } k=1\\&\frac{n(n-1)}2&\text{if }1<k<n/2\\&\frac{(n-1)(n-2)}2&\text{if }k=n/2 \end{aligned}\right.$} 
&$n-1$ 
&\tiny{$\left\{\begin{aligned}&n&\text{if } k=1\\&\frac{n(n-1)}2&\text{if }k>1\end{aligned}\right.$}
&\tiny{$\left\{\begin{aligned}&d_k^2&\text{if } k< n/2\\&d_k^2/2-1&\text{if }k=n/2\end{aligned}\right.$}&$d_k^2-1$\\
\hline
\end{tabular}} 
\end{center}
\end{table*}
\begin{table}[tb]
\begin{center}
\caption{The classification of fifteen configurations to eight DLA types. Checkmark denotes the corresponding coefficient is non-zero in the decomposition of the Hamiltonian.}
\label{tab:DLAcfg}
\resizebox{\linewidth}{!}{\renewcommand{\arraystretch}{1.3}\begin{tabular}{c|ccccc|c|c|c|c|c|c|cccc} 
\hline
$r$&\Checkmark&\Checkmark&\Checkmark&          &          &\Checkmark&          &          &          &          &\Checkmark&\Checkmark&\Checkmark&\Checkmark&          \\
$j$&\Checkmark&\Checkmark&          &\Checkmark&\Checkmark&          &\Checkmark&          &          &          &          &\Checkmark&\Checkmark&          &\Checkmark\\
$e$&\Checkmark&\Checkmark&\Checkmark&\Checkmark&\Checkmark&          &          &\Checkmark&          &\Checkmark&\Checkmark&          &          &          &          \\
$s$&\Checkmark&          &\Checkmark&          &\Checkmark&          &          &          &\Checkmark&\Checkmark&          &\Checkmark&          &\Checkmark&\Checkmark\\
\hline
DLA& \multicolumn{5}{c|}{\Romannum{1}}&\Romannum{2}&\Romannum{3}&\Romannum{4}&\Romannum{5}&\Romannum{6}&\Romannum{7}& \multicolumn{4}{c}{\Romannum{8}}\\
\hline
\end{tabular}} 
\end{center}
\end{table}

\end{proof}

% By calculating the DLA dimension, we explore the commutation relations between the four basis matrices, where there are 16 possible combinations (whether $e,r,j,s\neq 0$). These combinations are then classified into eight distinct DLA types, each with a unique dimensionality. Of the 16 combinations, only five yield full dimensionality in the FC configuration, meeting the conditions set forth in Theorem~\ref{thm:main}.

Given that achieving FC on most quantum processors remains a significant challenge, we also extend our analysis to a more realistic scenario involving nearest-neighbor (NN) connectivity. In this configuration, each physical qubit is only connected to its two adjacent qubits, forming a circular topology. The impact of this limited connectivity on the universality of the HWP ansatz is explored in Lemma~\ref{lemma:NN}. 
\begin{lemma}
\label{lemma:NN}
If the set of generators contains all four basis matrices on NN qubits, which are $\mathbf R_ {i, i + 1}$, $\mathbf J_{i, i + 1}$, $\mathbf E_ {i, i + 1}$, $\mathbf S_{i, i + 1}$, the dimension of DLA is $d_k^2$.
\end{lemma}
\begin{proof}
\begin{equation}\label{eq:JE}
[\mathbf J_{i,i+1},\mathbf E_{i+1,i+2}]=-\text i \mathbf R_{i,i+1}\otimes\sigma^z_{i+2}
\end{equation}
\begin{equation}
[\mathbf R_{i,i+1},\mathbf E_{i+1,i+2}]=\text i \mathbf J_{i,i+1}\otimes\sigma^z_{i+2}
\end{equation}
Then it can generate all $\mathbf R_ {i, i + 1} \otimes \sigma ^ z_ {i+2}$ and $\mathbf{J}_ {i, i + 1} \otimes \sigma ^ z_ {i+2} $, and
\begin{equation}
[\mathbf R_{i,i+1}\otimes\sigma^z_{i+2},\mathbf{R}_{i+1,i+2}]=\text i\mathbf J_{i,i+2}
\end{equation}
\begin{equation}
[\mathbf J_{i,i+1}\otimes\sigma^z_{i+2},\mathbf{R}_{i+1,i+2}]=-\text i \mathbf R_{i,i+2}
\end{equation}
Then it can generate all $\mathbf R_ {i, i+ 2} $ and $\mathbf J_ {i, i + 2}$, and so that all  $\mathbf R_ {ij} $, $\mathbf J_ {ij}$ and $\mathbf S_{ij}$ (since $ [\mathbf R_{ij},\mathbf J_{ij}]=-2i\mathbf S_{ij}$). With $[\mathbf J_{ik},[\mathbf J_{ij},\mathbf R_{jk}]]=2 \mathbf S_{ik}\otimes\sigma^z_j$, all the $\mathbf S_{ij}\otimes\sigma^z_k$ can be generated too. We notice that
\begin{equation}
\mathbf S_{ij}\otimes \sigma^z_{k} + \mathbf S_{jk}\otimes\sigma^z_i =- \mathbf E_{ij} + \mathbf E_{jk}
\end{equation}
Therefore, all $\mathbf E_{ij}$ can be generated by $\mathbf S_{i+1,i}\otimes \sigma^z_{j} + \mathbf S_{i,j}\otimes\sigma^z_{i+1} +\mathbf E_{i,i+1}=\mathbf E_{ij}$, so the DLA can be specified to type \Romannum{1} with the dimension $d_k^2$
\end{proof}

The lemma reveals that if the generator set contains all four basis matrices acting on NN qubits, the DLA dimension remains $d_k^2$. This result establishes a strong condition for NN connectivity, and we proceed by reducing the number of basis matrices in the generator set. Given that NN connectivity allows for fewer generators than FC, the DLA dimension for NN is either smaller than or equal to that for FC. Thus, we focus on verifying the five combinations (including the one in Lemma~\ref{lemma:NN}) from the two primary conditions in Theorem~\ref{thm:main}.
\begin{corollary}\label{coro:NN}
    For any $n$ and $k$, a two-qubit HWP gate is universal with NN connectivity if and only if the coefficients satisfy the following conditions:
    \begin{equation}
        \begin{aligned}
        \text{(1) }&e\neq0,j\neq0,r\neq0;\text{\quad (2) }e\neq0,j\neq0,s\neq0;\\
        \text{(3) }&e\neq0,r\neq0,s\neq0.
        \end{aligned}
        \label{eq:pro}
    \end{equation}
\end{corollary}
\begin{proof}
1) $r,\ j,\ e,\ s\neq 0$: Derived in Lemma~\ref{lemma:NN}.\\
2) $r,\ e,\ s\neq0,\ j=0$: $[\mathbf R_{i,i+1},\mathbf S_{i,i+1}]=\mathbf J_{i,i+1}$, leading to Lemma~\ref{lemma:NN}.\\
3) $r,\ j,\ e\neq0,\ s=0$: $[\mathbf R_{i,i+1},\mathbf J_{i,i+1}]=\mathbf S_{i,i+1}$, leading to Lemma~\ref{lemma:NN}.\\
4) $j,\ e,\ s\neq 0, r=0$: $[\mathbf J_{i,i+1},\mathbf S_{i,i+1}]=\mathbf R_{i,i+1}$, leading to Lemma~\ref{lemma:NN}.\\
5) $j,\ e\neq0,\ r=s=0$: From Table~\ref{tab:comm_basis}, 
\begin{equation}
\begin{aligned}
    [\mathbf{E}_{i,i+1},\mathbf{J}_{i+1,i+2}] &=\mathbf{R}_{i+1,i+2}\otimes\sigma^z_i\\
    [\mathbf{R}_{i,i+1}\otimes\sigma^z_{i+2},\mathbf{J}_{i+1,i+2}]&=-\mathbf{R}_{i+2,i}\\
    [\mathbf{J}_{i,i+1},\mathbf{E}_{i+1,i+2}]&=-\text{i} \mathbf{R}_{i,i+1}\otimes\sigma^z_{i+2}\\
    [\mathbf{R}_{i,i+2},\mathbf{E}_{i+2,i+3}]&=\text{i} \mathbf{J}_{i,i+2}\otimes\sigma^z_{i+3}\\
    [\mathbf{R}_{i,i+1}\otimes\sigma^z_{i+2},\mathbf{R}_{i,i+2}]&=\text{i}\mathbf{J}_{i,i+1}
\end{aligned}
\end{equation}
This set of generators can't generate basis elements $\mathbf R_{i,i+1}$ and $\mathbf{S}_{i,i+1}$, so it's specified to a different DLA type with 
\begin{equation}
\dim(\text{DLA})=\left\{\begin{aligned}&d_k^2&\text{if } k< n/2\\&d_k^2/2-1&\text{if }k=n/2\end{aligned}\right.
\end{equation}
\end{proof}

Corollary~\ref{coro:NN} indeed only rules out the combination $j,e\neq 0, r=s=0$. The above theorems offer top-down practical guidance for designing HWP ansätze and a robust framework for evaluating their effectiveness. While FC ensures maximum dimensionality, our findings show that even with the constraints of NN connectivity, universality can still be achieved under specific conditions. 

\begin{figure*}[t]
    \centering
    \includegraphics[width=\textwidth]{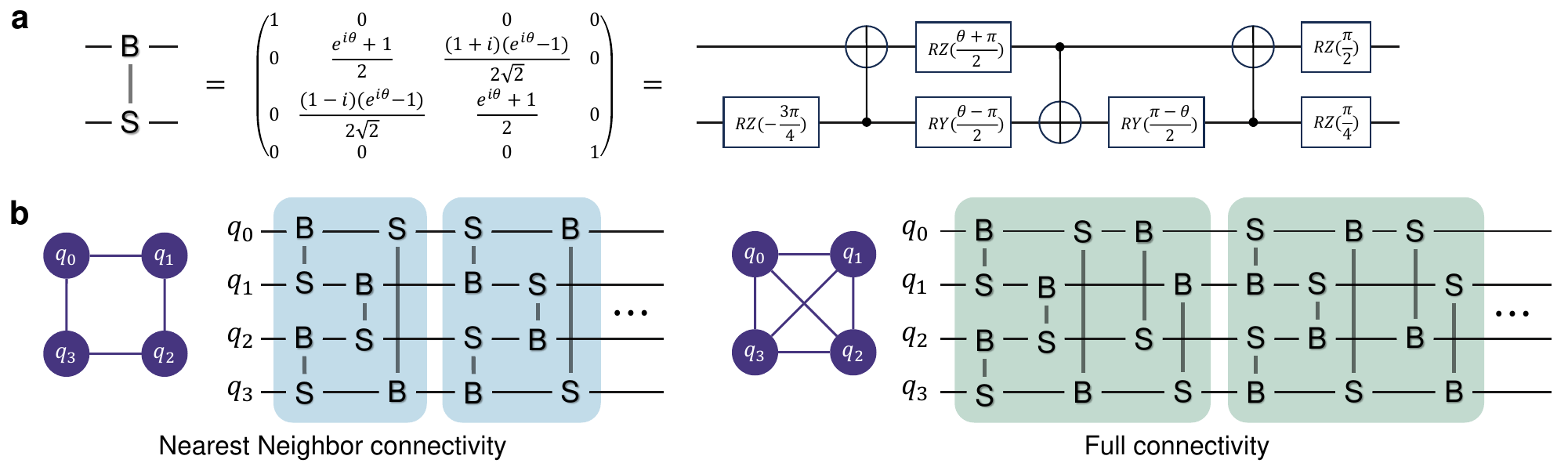}
    \caption{\textbf{Implementation details of the BS ansatz.} 
    \textbf{a.} A possible circuit implementation of the proposed BS gate. \textbf{b.} The implementation of the BS ansatz with both NN and FC connectivity. The physical qubit layout is depicted on the left. FC allows more generators per layer than NN connectivity. Additionally, certain HWP gates, e.g., the BS gate, are directional. To maximize expressivity, we alternate between layers containing BS gates and their reverse counterparts.}
    \label{fig:method}
\end{figure*}

\subsection{Theoretical Results of Trainability}
We now turn our attention to the trainability of the proposed HWP ansatz, which is critical due to the known barren plateaus problem in VQAs~\cite{mcclean2018barren,wang2021noise}. Typically, the gradients of the cost function decay exponentially with the number of qubits, scaling as $\mathcal{O}(1/2^n)$, since VQAs operate in the $2^n$-dimensional Hilbert space. However, as the HWP ansatz is confined to a $SU(d_k)$ subspace, we conjecture that the gradient decay should correspond to the subspace dimension, scaling as $\mathcal{O}(1/d_k)$. The following theorem provides the variance of the cost function gradients, using the framework established in~\cite{mcclean2018barren}:
\begin{theorem}
    Consider an $n$-qubit quantum ansatz operating in the subspace with Hamming weight  $k$. The variance of the cost function partial derivative is $Var_\theta[\partial_lC]\approx \frac{16k^2(n-k)^2}{n^4d_k}$.
    \label{theory:train}
\end{theorem}
\begin{proof}
Consider the partial derivative of the cost function $C$ with respect to the parameters $\theta$. For some parameter $\theta_l$ in the $l$-th HWP gate $\mathbf{H}_l$, we have:
\begin{equation}
    \begin{aligned}
        \partial_{l} C(\theta)&=\partial_l\Big(\Tr\big(\mathbf{U}(\theta)\rho \mathbf{U}(\theta)^{\dagger}O\big)\Big),
    \end{aligned}
\end{equation}
where $\rho$ is the input state, $O$ is the observable to measure. We split the whole ansatz to two parts with $\mathbf{U}_-$ denotes the unitary matrix of the circuit before $\mathbf{H}_l$ and $\mathbf{U}_+$ denotes the unitary matrix after gate $l$.
\begin{equation}
    \begin{aligned}
        \partial_{l} C(\theta)&=\partial_l\Big(\Tr\big(\mathbf{U}_-\rho \mathbf{U}_-^{\dagger} O_+\big)\Big) \\
        &=\text{i}\Tr\big(\mathbf{U}_-\rho \mathbf{U}_-^{\dagger} [\mathbf{H}_l,O_+]\big),
    \end{aligned}
\end{equation}
where $[\cdot,\cdot]$ denotes the commutator of two matrices. The variance of the partial derivative is thus given by
{\footnotesize\begin{equation}
    \begin{split}
        Var_{\theta}[\partial_l C] &= \int_{\mathbf{U}_+}\mathrm{d}\mathbf{U}_+ \int_{\mathbf{U}_-}\mathrm{d}\mathbf{U}_- \big(\partial_l C(\theta)\big)^2\\
        &= \int_{\mathbf{U}_+}\mathrm{d}\mathbf{U}_+ \int_{\mathbf{U}_-}\mathrm{d}\mathbf{U}_- \Big(\text{i}\Tr\big(\mathbf{U}_-\rho \mathbf{U}_-^{\dagger} [\mathbf{H}_l,O_+]\big)\Big)^2\\
        &= -\int_{\mathbf{U}_+}\mathrm{d}\mathbf{U}_+ \Bigg( \frac{\Tr(\rho^2)\Tr\big([H_l,O_+]^2\big)}{d_k^2-1} \\
        &\quad\quad - \frac{\Tr^2(\rho)\Tr\big([H_l,O_+]^2\big)}{d_k(d_k^2-1)} \Bigg)\\
        &= -\int_{\mathbf{U}_+}\mathrm{d}\mathbf{U}_+ \Bigg(\Tr\big( [\mathbf{H}_l,O_+]^2\big) \frac{d_k\Tr(\rho^2)-\Tr^2(\rho)}{d_k(d_k^2-1)}\Bigg)\\
        &= -\frac{d_k\Tr(\rho^2)-\Tr^2(\rho)}{d_k(d_k^2-1)}\int_{\mathbf{U}_+}\mathrm{d}\mathbf{U}_+ \Tr\big( [\mathbf{H}_l,O_+]^2\big) .
    \end{split}
\end{equation}}
As long as the initial state is the $SU(d_k)$ subspace, we have $\Tr(\rho)=1$ and $\Tr(\rho^2)=1$.
{\small\begin{equation}\label{eq:C99}
    \begin{split}
        Var_\theta[\partial_lC] &= -\frac{1}{d_k(d_k+1)}\int_{\mathbf{U}_+}\mathrm{d}\mathbf{U}_+ \Tr\big( [H_l,O_+]^2\big)\\
        &= -\frac{2}{d_k(d_k+1)}\Big(\Tr(\mathbf{H}_lO_+\mathbf{H}_lO_+) \\
        &\quad\quad -\Tr(\mathbf{H}_l\mathbf{H}_lO_+O_+)\Big)\\
        &= -\frac{2}{d_k(d_k+1)}\Bigg( \frac{\Tr(\mathbf{H}^2_l)\Tr^2(O)}{d_k^2-1} \\
        &\quad\quad -\frac{\Tr(\mathbf{H}^2_l)\Tr(O^2)}{d_k(d_k^2-1)}-\frac{\Tr(\mathbf{H}^2_l)\Tr(O^2)}{d_k}\Bigg)\\
        &= -\frac{2\Tr(\mathbf{H}_l^2)}{d_k(d_k+1)}\Bigg(\frac{\Tr^2(O)-d_k \Tr(O^2)}{d_k^2-1}\Bigg),
    \end{split}
\end{equation}}
where $\Tr(\mathbf{H}_l^2)=2\binom{n-2}{k-1}=\frac{2k(n-k)}{n(n-1)}d_k$. Without loss of generality, we set the observable $O$ as $Z_0$ since other observables will also hold with the same magnitude. Thus, $\Tr(O)=\frac{d_k(n-2k)}{n}$ and $\Tr(O^2)=d_k$. Substitute this back into Equation~\ref{eq:C99} and we get
{\small\begin{equation}
    \begin{split}
        Var_{\theta}[\partial_l C]&=-\frac{2}{d_k(d_k+1)}\cdot\frac{2k(n-k)d_k}{n(n-1)}\cdot\frac{\frac{d_k^2(n-2k)^2}{n^2}-d_k^2}{d_k^2-1}\\
        &=\frac{4k(n-k)}{(d_k+1)n(n-1)}\cdot \frac{d_k^2(n^2-(n-2k)^2)}{(d_k^2-1)n^2}\\
        &=\frac{4k(n-k)}{(d_k+1)n(n-1)}\cdot\frac{d_k^2(4nk-4k^2)}{(d_k^2-1)n^2}\\
        &=\frac{16k^2(n-k)^2 d_k^2}{(d_k+1)n^3(n-1)(d_k^2-1)}\\
        &\approx \frac{16 k^2(n-k)^2}{n^4 d_k}.
    \end{split}
\end{equation}}
\end{proof}

Theorem~\ref{theory:train} supports our conjecture that the trainability of the circuit is governed by the subspace dimensionality $d_k$ rather than the exponentially large qubit number $2^n$. Upon further analysis, we find that if $k=1$, then $Var_\theta[\partial_lC]\approx \frac{16}{n^3}$. Conversely, when $k=\frac{n}{2}$, $Var_\theta[\partial_lC]\approx \binom{n}{n/2}^{-1}$, which is still approximate to exponentially small.

While this indicates that barren plateaus (BPs) are not entirely eliminated in the $k\approx n/2$ regime, it is imperative to contextualize this $\mathcal{O}(1/d_k)$ scaling within the fundamental expressivity-trainability tradeoff~\cite{cerezo2021cost, larocca2023theory}. Compared to unconstrained HEAs, whose gradients vanish as $\mathcal{O}(1/2^n)$, our HWP ansatz exponentially suppresses the severity of BPs by exclusively confining the exploration to the physically valid manifold. Furthermore, compared to Hamiltonian-variational methods that might empirically exhibit better trainability, our $\mathcal{O}(1/d_k)$ bound represents the optimal theoretical limit for any truncation-free, subspace-universal ansatz. If a quantum circuit were to completely avoid exponential gradient vanishing (e.g., exhibiting at-worst polynomial variance), its DLA would necessarily be polynomially bounded. According to well-established classical simulability theories~\cite{vidal2003efficient,somma2006efficient,ragone2024lie}, such highly restricted circuits fundamentally lack the representational capacity to capture highly entangled, strongly correlated many-body states, rendering them efficiently simulable on classical computers.

% As highlighted in~\cite{cerezo2021cost}, there exists a tradeoff between expressivity and trainability: more expressive ansätze might offer greater representational power, but they also suffer from diminished trainability, as optimization becomes increasingly difficult due to flat cost landscapes. Therefore, designing ansätze that operate within a restricted subspace, such as the HWP ansatz, may offer a potential solution to this tradeoff.

\subsection{Universal 2-Local HWP Ansatz}
In this section, we detail the construction of the universal HWP ansatz, guided by the algebraic frameworks established in preceding sections. The necessary and sufficient conditions derived in Theorem \ref{thm:main} and Corollary \ref{coro:NN} provide a rigorous blueprint for designing universal operators under both FC and NN connectivity. Notably, none of the heuristic gates summarized in Figure \ref{fig:HWPgates} satisfy these strict universality conditions. Typical operators derived from physical hopping terms are composed of a single basis matrix; for instance, the Givens Rotation (GR) widely used in electronic structure simulations relies solely on $\mathbf{J}_{ij}$, while the XY-interaction in the Fermi-Hubbard model relies entirely on $\mathbf{R}_{ij}$. While these operators efficiently preserve the Hamming weight, they fundamentally restrict the system's evolution to isolated, low-dimensional invariant sub-manifolds (e.g., restricting to the real-valued orthogonal group $SO(d_k)$ in the case of GR). This inherent algebraic limitation strictly prevents these ansätze from spanning the complete special unitary group $SU(d_k)$ required for absolute subspace universality.

To break these invariant sub-algebras, an optimal operator must intrinsically couple the amplitude exchange with localized complex phase modulations. Guided by Condition (1) in Corollary~\ref{coro:NN}  with $e=r=j=1$ and incorporating normalization coefficients, we construct the following operator that explicitly resolves this deficit. This operator has never been discovered in any Hamiltonian before, and we denote this operator as the BS gate
\begin{equation}
    \mathbf{H}_{BS}= \left(\begin{array}{cccc}
        0 & 0 & 0 & 0\\
        0 & \frac{1}{2} & \frac{1+\text{i}}{2\sqrt{2}} & 0\\
        0 & \frac{1-\text{i}}{2\sqrt{2}} & \frac{1}{2}& 0\\
        0 & 0 & 0 & 0
    \end{array}\right),
\end{equation}
with $\mathbf{H}_{BS}=\mathbf{H}_{BS}^2$. We obtain the unitary matrix of the BS gate by Taylor expansion:
{\begin{equation}    
    \mathbf{U}_{BS}(\theta)=e^{\text{i}\theta \mathbf{H}_{BS}}= \left(\begin{array}{cccc}
        1 & 0 & 0 & 0\\
        0 & \frac{(e^{\text{i}\theta}+1)}{2} & \frac{(1+\text{i})(e^{\text{i}\theta}-1)}{2\sqrt{2}} & 0\\
        0 & \frac{(1-\text{i})(e^{\text{i}\theta}-1)}{2\sqrt{2}} & \frac{(e^{\text{i}\theta}+1)}{2} & 0\\
        0 & 0 & 0 & 1
    \end{array}\right),
\end{equation}}
where $e^{\text{i}\theta}=\cos(\theta)+\text{i}\sin(\theta)$. 
A hardware-efficient decomposition of the BS gate is illustrated in Figure~\ref{fig:method}a, requiring three CNOT gates alongside parameterized single-qubit rotations. This hardware efficiency is defined in explicit contrast to Hamiltonian variational ansätze like UCCSD, where operators with order higher than two are required. Specifically, the 4-local double excitation operators in UCCSD are intrinsically non-local across the qubit register. Implementing these complex operators on near-term hardware with restricted physical connectivity necessitates prohibitive SWAP gate overheads. In contrast, the proposed BS ansatz is strictly 2-local and inherently tailored for NN connectivity.

Building upon this operator, Figure~\ref{fig:method}b presents two ansatz architectures tailored to both NN and FC physical qubit connectivities. While advanced quantum platforms, such as trapped ions or reconfigurable neutral atom arrays, can natively support or dynamically simulate long-range interactions, executing dense all-to-all connectivity in deep parameterized circuits still imposes severe physical overheads. Consequently, we deliberately constrain our primary focus to the NN architecture. Establishing theoretical universality and empirical expressivity under the strict topological bottleneck of 1D NN connectivity guarantees that the proposed BS ansatz remains universally executable and highly resource-efficient across all major quantum physical backends, including rigidly constrained superconducting circuits. Furthermore, the mathematically derived BS gate is inherently directional (i.e., $\mathbf{U}_{ij}\neq\mathbf{U}_{ji}$). To systematically break this directional bias and maximize the algebraic expressivity across the qubit lattice, we construct the ansatz by alternating layers of forward BS gates with their reversed counterparts. This alternating topological arrangement is strictly required to ensure that the generated DLA achieves the maximum dimensionality required for absolute subspace universality.

A significant conceptual advancement of the proposed BS ansatz is its complete decoupling from the target problem Hamiltonian. Conventional VQE circuit designs are largely phenomenological: the UCCSD ansatz is tightly tied to the molecular electron-excitation Hamiltonian, whereas the EHV ansatz is tailored to the hopping and on-site terms of the Fermi-Hubbard model. 
Consequently, these heuristic ansätze must be laboriously re-engineered for different physical systems, and their expressivity is inherently bounded by the truncated terms of the specific Hamiltonian. 
In contrast, the BS ansatz is derived purely from the underlying mathematical constraint, the Hamming Weight preservation, independent of any specific interaction coefficients. 
This represents a paradigm shift from Hamiltonian-driven circuit design to Constraint-driven algebraic compilation. 
By mathematically guaranteeing complete controllability over the constrained subspace, the exact same untailored BS ansatz architecture can serve as a universal solver for any fermionic many-body system, seamlessly bridging disparate domains such as molecular electronic structures and condensed-matter models.

\begin{figure*}[t]
    \centering
    \includegraphics[width=0.95\linewidth]{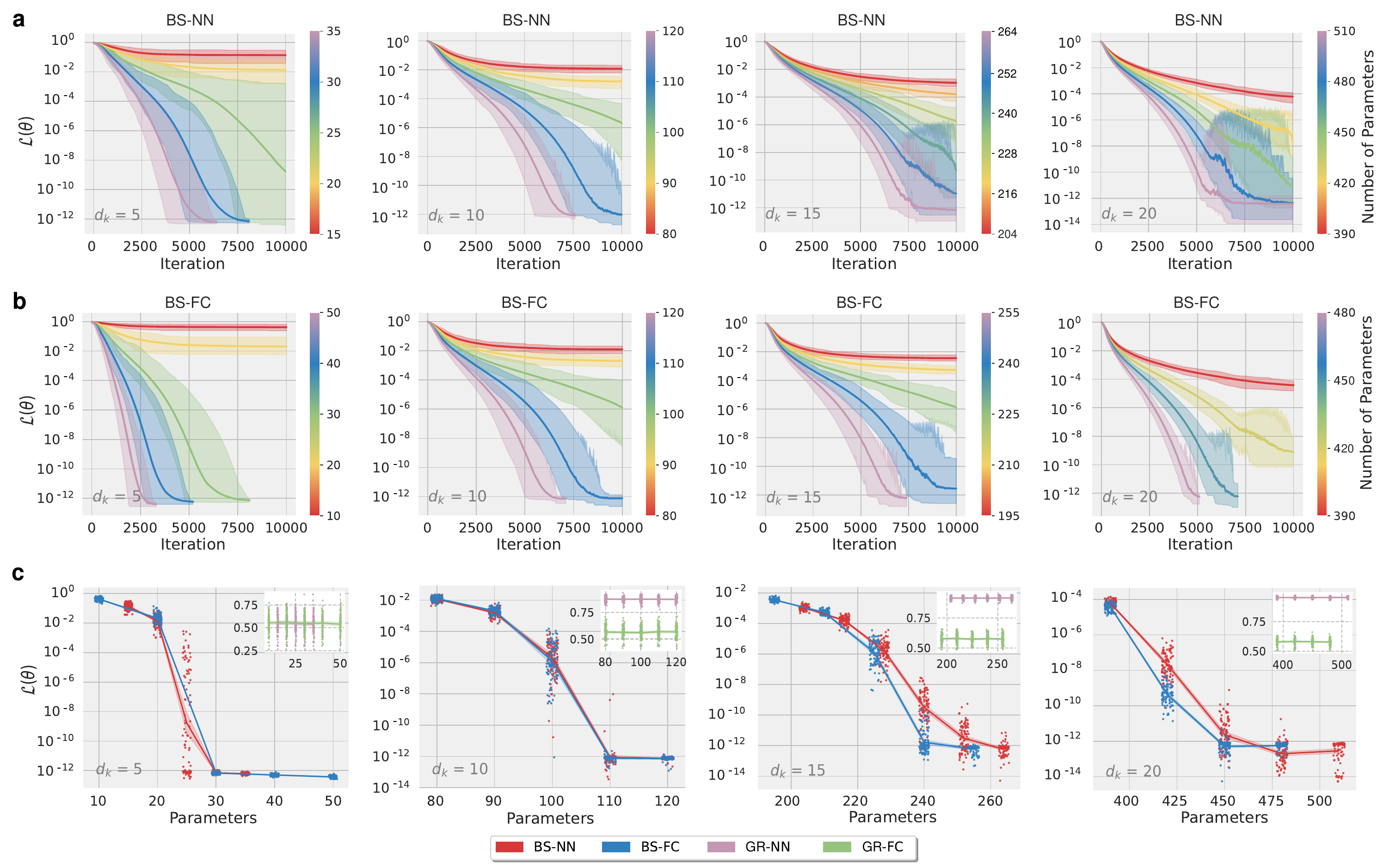}
    \caption{\textbf{Results for unitary approximation.} We iterate through all the cases for $d_k=\{\tbinom{5}{1},\tbinom{5}{2},\tbinom{6}{2},\tbinom{6}{3}\}=\{5,10,15,20\}$, with both GR and BS gates for NN and FC connectivity. For each $d_k$, 100 unitary matrices are randomly sampled based on Haar measure~\cite{zyczkowski1994random,mezzadri2006generate}. \textbf{a}. The training curves for BS with NN connectivity. \textbf{b}. The training curves for BS with FC connectivity. \textbf{c}. The loss function is plotted versus the number of parameters. Both BS-NN and BS-FC show a similar decreasing pattern with the number of parameters needed for exact approximation at around $d_k^2$. Inset shows the results for GR-NN and GR-FC, with neither method getting close to $\mathcal{L}(\theta)=0$.}
    \label{fig:unitary_approximation}
\end{figure*}

\section{Experiments}
Experiments were conducted on a machine with 2TB memory, 100 cores Intel Xeon Platinum 8480+ CPU, and 8 GPUs (Nvidia A100). All numerical simulations presented in this work were implemented using the MindSpore Quantum framework~\cite{xu2024mindspore}.

\begin{figure*}[t]
    \centering
    \includegraphics[width=0.9\linewidth]{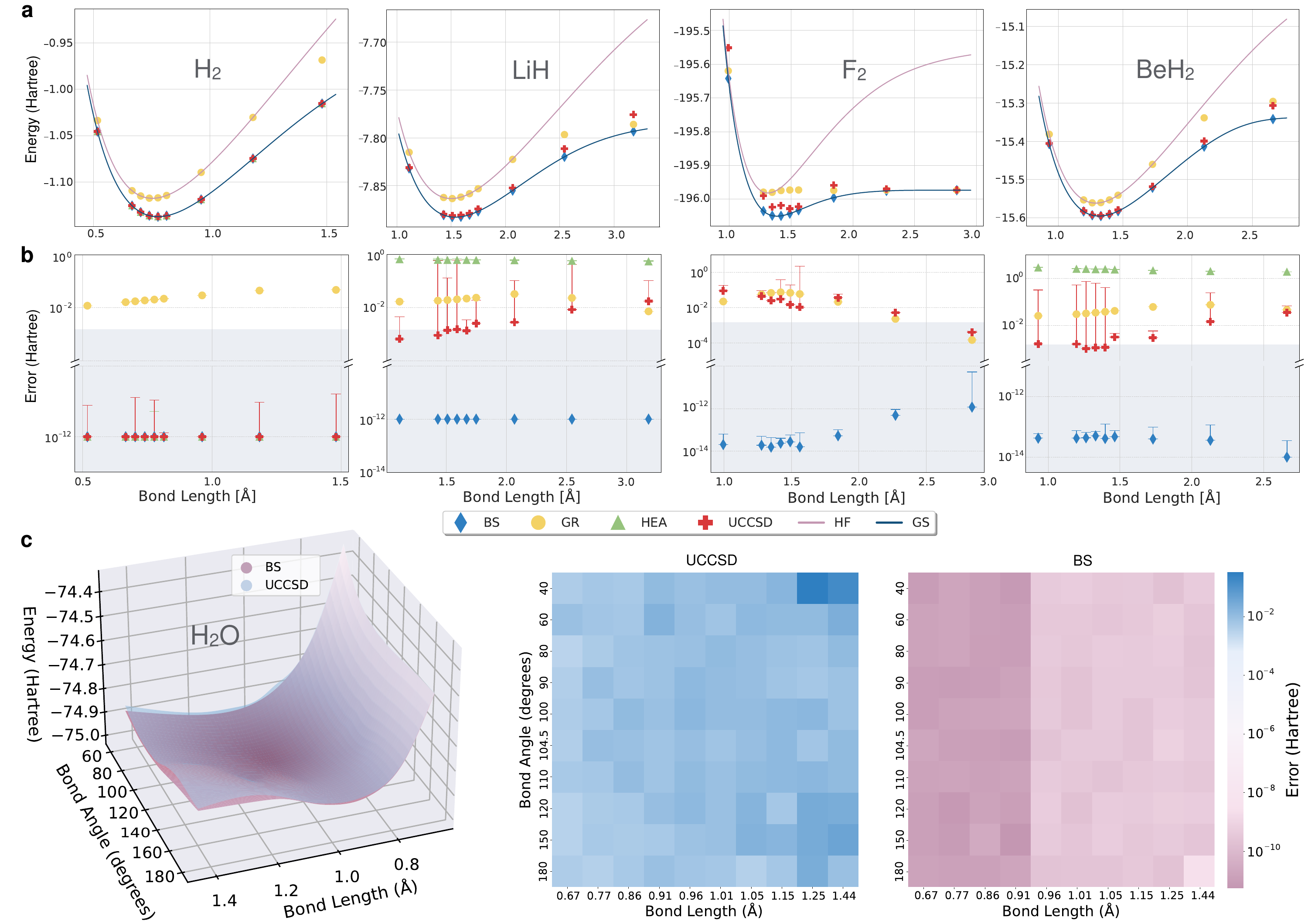}
    \caption{\textbf{Results for simulating molecular electronic structures.} \textbf{a}. Potential energy curves of four molecules w.r.t. the bond length, with the corresponding absolute errors compared to the exact energy shown in \textbf{b}. HF stands for the energy with Hartree-Fock state, and GS stands for the exact energy. The grey region shows results within chemical accuracy (less than 1.6 milli-hartrees). \textbf{c}. Potential energy surface of H$_2$O, and the energy error of $100$ sampled molecule structures for both UCCSD and BS. All data points are a minimum of 10 random seeds, with error bars indicating the range from minimum to maximum.}
    \label{fig:UCC}
\end{figure*}
\subsection{Numerical Results for Unitary Approximation}
\subsubsection{Problem Definition}
To verify the theoretical results on the universality of different gates and connectivity, we provide numerical results on the unitary approximation problem~\cite{nielsen2002quantum,chong2017programming,khatri2019quantum}. 
The unitary approximation aims to solve the problem of whether $\{\mathbf{U}(\bm{\theta})\}_{\bm{\theta}}$ is equal to $\mathcal{SU}(N)$ (see lemma.~\ref{lemma:control}). For a target unitary matrix $\hat{\mathbf{U}}$ in $d_k$-dimensional HWP subspace, the loss function for unitary approximation is
\begin{equation}\label{eq:UAloss}
    \mathcal{L}_{UA}(\bm{\theta})=1-|\Tr\big(\hat{\mathbf{U}}^\dagger \mathbf{U}(\bm{\theta})\big)|^2/d_k^2.
\end{equation}

To illustrate ansatz with DLA dimension less than $d_k^2$ can not approximate an arbitrary unitary matrix to any desired accuracy, we also provide results for the Givens Rotations (GR)~\cite{wecker2015solving,jiang2018quantum}, 
\begin{equation}
    \mathbf{H}_{GR}=\left(\begin{array}{cccc}
        0 & 0 & 0 & 0\\
        0 & 0 & -\text{i} & 0\\
        0 & \text{i} & 0 & 0\\
        0 & 0 & 0 & 0
    \end{array}\right),
\end{equation}
which is widely used in Fermionic system simulation. Figure~\ref{fig:unitary_approximation} shows the results of minimizing the loss in Equation~\ref{eq:UAloss} for four cases with HWP subspace dimension $d_k=\{\tbinom{5}{1},\tbinom{5}{2},\tbinom{6}{2},\tbinom{6}{3}\}=\{5,10,15,20\}$. 

\subsubsection{Results}
Our numerical results demonstrate that the theoretically derived BS gate can approximate a target unitary matrix with arbitrary precision, achieving a loss function value as low as $1 \times 10^{-12}$ (see Figure~\ref{fig:unitary_approximation}a and b). In Figure~\ref{fig:unitary_approximation}c, we plot the final approximation loss for 100 randomly sampled unitary matrices across various parameter counts. The training dynamics of the NN and FC configurations exhibit highly similar convergence trajectories. Because the number of constituent gates per layer inherently differs between the NN and FC topologies, their total parameter counts cannot be perfectly aligned unless they share a common multiple. Nevertheless, the overall macroscopic trends remain remarkably consistent. 

Notably, the loss distribution for the NN configuration exhibits a slightly wider variance compared to the FC configuration. This subtle discrepancy suggests that while NN and FC share identical algebraic expressivity (equivalent DLA dimensions), their underlying optimization landscapes may possess distinct topological features and curvature properties. In stark contrast, other commonly used HWP operators, such as the GR gate, fundamentally fail to approximate arbitrary unitary matrices. As depicted in the inset of Figure~\ref{fig:unitary_approximation}c, the approximation loss for the GR ansatz stagnates near 0.5 or 0.75, showing no notable improvement even in the ultra-deep, extensively parameterized regime. This limitation indicates an inherent inability to span the complete target subspace, a phenomenon theoretically analogous to the truncation error observed in the coupled cluster theorem when restricted solely to single-excitation operators. Ultimately, these findings corroborate our theoretical predictions: achieving a truncation-free, exact approximation does not strictly necessitate complex high-order non-local terms. Instead, it can be efficiently realized using mathematically validated 2-local HWP operators.

Furthermore, Figure~\ref{fig:unitary_approximation}c illustrates the critical number of parameters (and consequently, circuit depth) required for the ansatz to successfully converge to this exact approximation. The sharp downward transition observed in the loss curves evokes a well-documented phenomenon in both classical and quantum machine learning, e.g., overparameterization~\cite{allen2019convergence}, in which spurious local minima vanish and the optimization landscape becomes highly navigable. Our empirical results demonstrate that an analogous computational phase transition occurs within the constrained HWP quantum subspace. Specifically, we validate the quantum overparameterization theory~\cite{larocca2023theory}, confirming that the critical parameter threshold required to guarantee trainability is intrinsically proportional to the DLA dimension of the ansatz, which, in our framework, scales strictly as $d_k^2$. By confining this overparameterization bound to the dimension of the symmetry-preserved subspace rather than the exponentially large full Hilbert space, our approach significantly suppresses the resource overhead required to escape barren plateaus. These insights fundamentally bridge quantum control theory with classical optimization dynamics, offering a theoretically sound blueprint for designing highly expressive and trainable HWP ansätze for practical quantum simulation tasks.

\begin{figure}[t]
    \centering
    \includegraphics[width=0.95\linewidth]{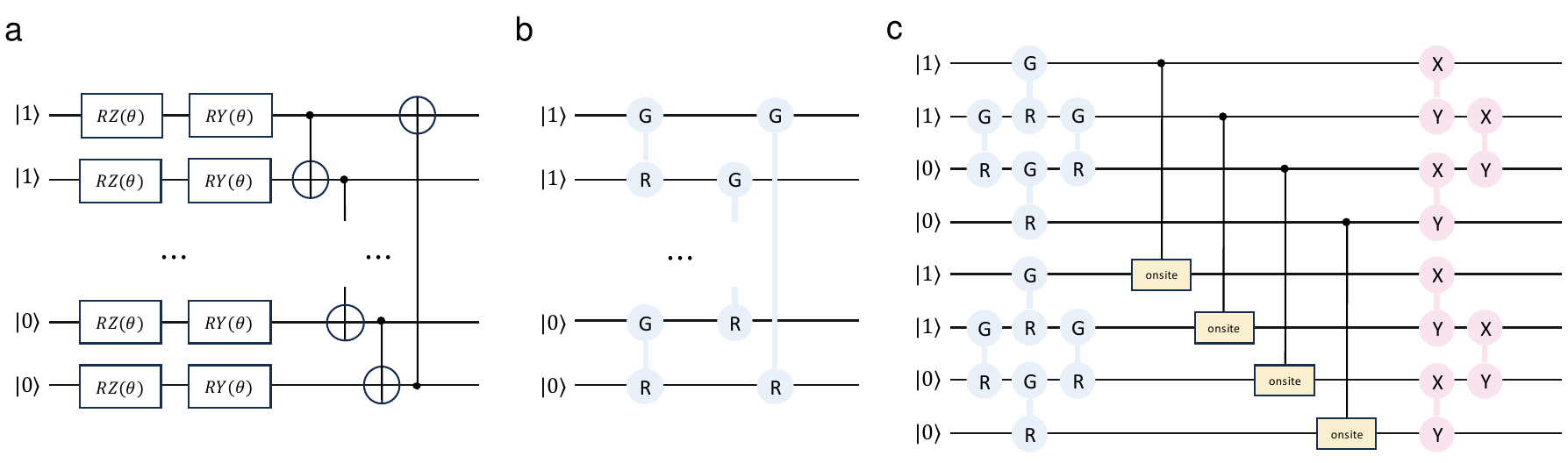}
    \caption{\textbf{Ansatz of Baseline Methods.} \textbf{a}. One layer of HEA ansatz utilizing parameterized rotation Y and Z gates and CNOT gates. The first $k$ (number of electrons) qubits are initialized to $\ket{1}$ as the HF state. \textbf{b}. One layer of GR ansatz with $n$ (number of orbitals) GR gates connecting neighboring qubits. \textbf{c}. An example for EHV ansatz for 1*4 sites Fermi-Hubbard model at half filling. Top four qubits for spin-ups and the rest for spin-downs. The GR gates are utilized for initialization, and the EHV ansatz requires repeatedly applying the onsite layer and XY-interaction layer.} 
    \label{fig:baseline}
\end{figure}
\begin{table}[t]
    \centering
    \caption{\textbf{Statistics of molecules.} $n$ and $k$ are the numbers of orbitals and electrons, respectively.}
    \label{tab:mole-stat}
    \resizebox{0.9\linewidth}{!}{\begin{tabular}{l | ccccc}
        \hline
        Molecules &  H$_2$  & LiH & H$_2$O & BeH$_2$ & F$_2$\\ 
        \hline
        $n$   & 4 &  12 & 14 & 14 & 20 \\
        $k$ & 2   & 4 & 10 & 6 & 18\\
        $d_k$ & 6  & 495 & 1001 & 3003 & 190\\
        \hline
        \end{tabular}}
\end{table}
\subsection{Numerical Results for Solving Electronic Structures}
\subsubsection{Probelm Definition}
We then leverage the proposed BS ansatz to address the challenges of quantum chemistry simulation~\cite{kandala2017hardware,anand2022sukin,guo2024experimental}, focusing on molecular Hamiltonians within the Born-Oppenheimer approximation and expressed in second quantization. These Hamiltonians, obtained from the Python package PySCF~\cite{sun2018pyscf}, are transformed into qubit form using the minimal basis set STO-3G and the Jordan-Wigner transformation. The Hamiltonian in second quantization is given by:
\begin{equation}
    \mathcal{H}=\sum_{i,j}h_{ij}\hat{a}_j^\dagger\hat{a}_i+\frac{1}{2}\sum_{i,j,p,q}g_{ijpq}\hat{a}_j^\dagger\hat{a}_q^\dagger\hat{a}_p\hat{a}_i,
\end{equation}
where $\hat{a}_i=\begin{pmatrix}\begin{smallmatrix}
    0&0\\1&0
\end{smallmatrix}\end{pmatrix}$ denotes the annihilation operator on qubit $i$, corresponding to the $i$-th molecular orbital in the active space. The coefficients $h_{ij}$ and $g_{ijpq}$ denote the one- and two-electron integrals, respectively. A well-established solving approach is the VQE method, with the chemically inspired unitary coupled cluster (UCC) ansatz~\cite{romero1701strategies,anand2022sukin,bartlett1989alternative}.
The coupled cluster operator, in second quantization, is defined as $\hat{T}=\hat{T}_1+\hat{T}_2+\cdots+\hat{T}_v$ with the single and double excitation operators as
\begin{equation}
    \begin{aligned}
        \hat{T}_1=\sum_{i,j}\hat{a}_j^\dagger\hat{a}_i,\quad\quad
        \hat{T}_2=\sum_{i,j,p,q}\hat{a}_j^\dagger\hat{a}_q^\dagger\hat{a}_p\hat{a}_i.
    \end{aligned}
\end{equation}
For a given reference initial state $\ket{\psi_0}$, the UCC ansatz wave function is given by $\ket{\psi}=e^{\hat{T}-\hat{T}^\dagger}\ket{\psi_0}$~\cite{bartlett1989alternative}, where $T-T^\dagger$ is an anti-Hermitian operator which makes it suitable for quantum computers since the exponential of an anti-Hermitian operator is a unitary operator. The Hamiltonian for the single excitation term is 
\begin{equation}
    \mathcal{H}_{single}=\frac{1}{\rm i}(\hat{T}-\hat{T}^\dagger)=\frac{1}{\rm i}\begin{pmatrix}\begin{smallmatrix}
        0 & 0 & 0 & 0\\
        0 & 0 & 1 & 0\\
        0 & -1 & 0 & 0\\
        0 & 0 & 0 & 0
    \end{smallmatrix}\end{pmatrix}=\begin{pmatrix}\begin{smallmatrix}
        0 & 0 & 0 & 0\\
        0 & 0 & -\text{i} & 0\\
        0 & \text{i} & 0 & 0\\
        0 & 0 & 0 & 0
    \end{smallmatrix}\end{pmatrix},
\end{equation}
which is indeed the GR we utilize as the baseline in unitary approximation. The Hamiltonian for the double excitation operator, on the other hand, is more complex, involving a $16\times 16$ matrix. This increased complexity significantly challenges implementation on quantum devices. 

\begin{figure}[t]
    \centering
    \includegraphics[width=0.9\linewidth]{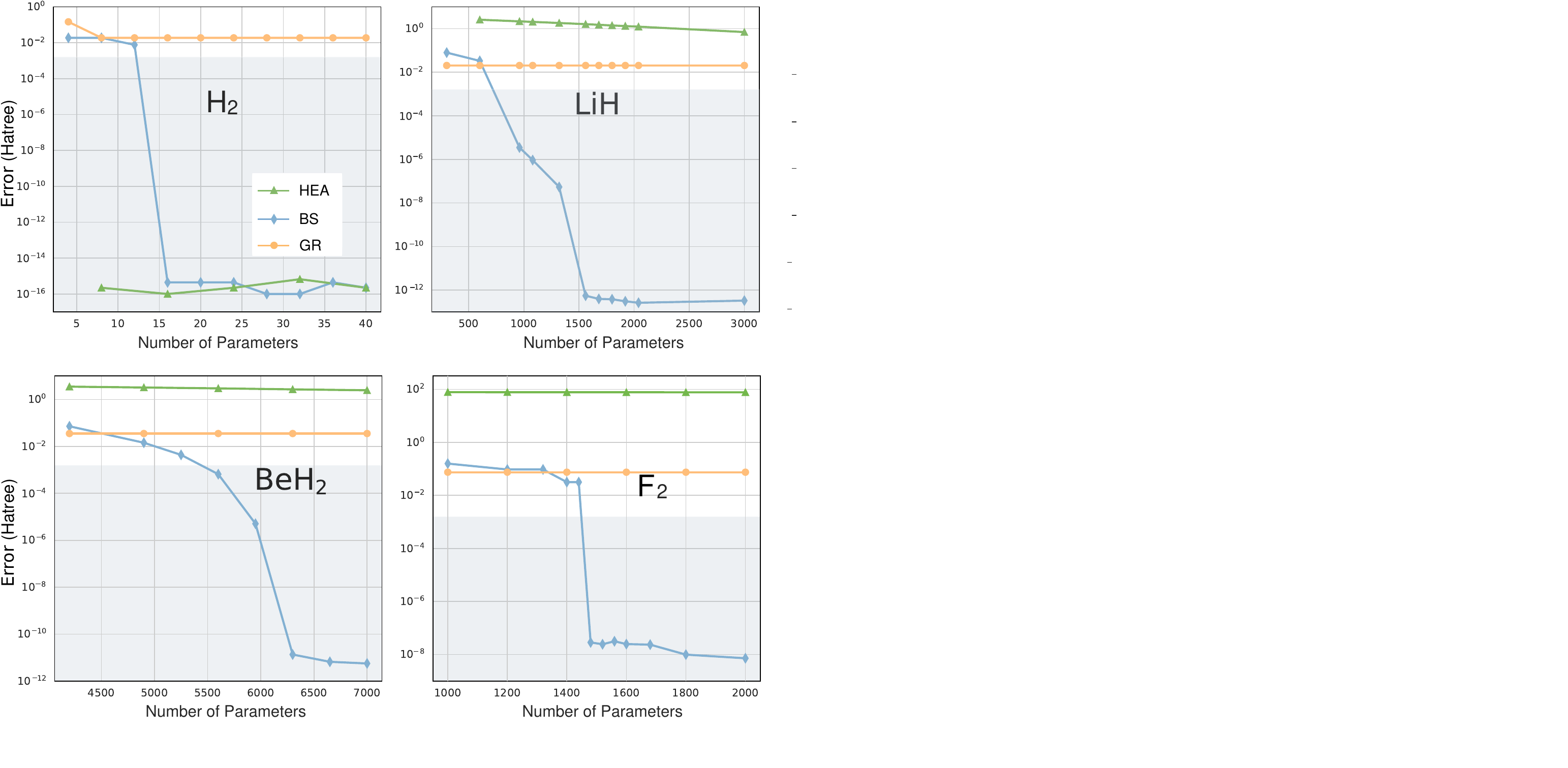}
    \caption{\textbf{Overparameterization for simulating molecular electronic structures.} The energy error w.r.t. the number of parameters. Grey area indicating error within Chemical Accuracy.}
    \label{fig:apxUCC}
\end{figure}

We conduct numerical experiments on five molecules with detailed statistics in Table~\ref{tab:mole-stat}. For H$_2$, LiH, F$_2$, and BeH$_2$, the data points across varying bond lengths are obtained by iterating over the coefficients in the set $\{0.7,0.9,0.95,1,1.05,1.1,1.3,1.6,2\}$ time the minimum-energy bond length. This setting focuses on both the performance near the minimum-energy point and at extended bond lengths. The maximum number of iterations is capped at 10,000, and all the results are a minimum of 10 random seeds, with the error bar indicating the range of 10 seeds. Since the number of layers is kept identical across all bond lengths, the results show that increasing bond length leads to greater variance in the BS ansatz, suggesting that the problem becomes more challenging as bond lengths increase.
For water, we explore the potential energy surface with respect to both bond length and bond angle. The bond length coefficients are sampled from the set $\{0.7,0.8,0.9,0.95,1,1.05,1.1,1.2,1.3,1.5\}$,  differing from the other molecules to achieve a more evenly distributed bond lengths. The bond angles are selected from $\{40^\circ,60^\circ,80^\circ,90^\circ,100^\circ,104.5^\circ,110^\circ,120^\circ,150^\circ,180^\circ\}$.
To better illustrate the efficiency of the proposed BS ansatz, we select HEA~\citep{kandala2017hardware} and UCCSD~\citep{romero1701strategies} ansatz as the baselines (implementation details in Figure~\ref{fig:baseline}). We also involve GR, representing the single excitation term of the UCC model~\cite{romero1701strategies}. Both BS and GR are under NN connectivity.

\subsubsection{Results}
As illustrated in Figure~\ref{fig:UCC}, the HEA (green triangles) exhibits the poorest performance due to its unconstrained search within the full Hilbert space, which leads to barren plateaus and symmetry-breaking solutions. The GR ansatz (yellow squares), while symmetry-preserving, is restricted to one-body interactions and fails to capture electron correlations essential for chemical accuracy.

The most significant insight emerges from the comparison between UCCSD and our BS ansatz. Contrary to the conventional wisdom in computational chemistry, which posits that high precision requires high-rank excitation operators (e.g., triples or quadruples), 2-local operators are believed insufficient for describing complex correlation. However, our results demonstrate a fundamental shift in perspective under the qubit encoding. 
While UCCSD (red crosses) relies on a chemically inspired but mathematically truncated operator pool (Singles and Doubles), the problem on a quantum circuit transforms into one of state reachability within a constrained subspace. Although constructed from simple 2-local gates, the BS ansatz acts as a universal building block. Through the lens of DLA, we prove that a sufficiently deep BS ansatz can generate the complete algebra of the HWP subspace. This effectively achieves a "truncation-free" approximation, allowing it to access any quantum state within the subspace—equivalent to Full Configuration Interaction (FCI)~\cite{knowles1984new}, despite using only hardware-efficient 2-local interactions.

This theoretical advantage translates directly into the numerical superiority shown in Figure~\ref{fig:UCC}. The BS ansatz (blue diamonds) consistently maintains energy errors below $1 \times 10^{-10}$ Hartree across all bond lengths, including the strongly correlated regions where the truncated UCCSD falters (errors $\sim 10^{-3}$ Hartree). This confirms that subspace universality, rather than the explicit inclusion of high-order physical operators, is the decisive factor for accuracy in quantum simulation.

In Figure~\ref{fig:apxUCC}, we further examine the impact of the number of parameters. Both the HEA and GR ansatz exhibit a consistent trend, indicating that they have reached the maximum expressibility of the ansatz but remain unable to approximate the ground state with sufficient accuracy. Notably, our findings reveal that the number of parameters required to reach overparameterization for molecular electronic structures is significantly smaller than expected, diverging from the anticipated $d_k^2$. For an error margin close to chemical accuracy, fewer than $2 \times d_k$ parameters are sufficient. This observation is particularly impactful, as it suggests that solving a VQE for Fermionic system simulation may require a number of parameters that scales linearly with $d_k$, offering valuable insights into the potential advantages of the VQE algorithm on intermediate-scale quantum processors.
% Our experiments demonstrate that the ansatz constructed from BS gates achieves an error level of $1\times 10^{-10}$Ha for molecular ground state energies, consistent across all tested bond lengths and bond angles. This precision significantly exceeds that of existing VQE methodologies. Furthermore, in our unitary approximation task, we established that any unitary matrix within the HWP subspace can be accurately approximated by our proposed BS-gate ansatz. Consequently, this implies that any state within the HWP subspace is accessible through our ansatz, effectively rendering it a quantum circuit-based implementation of the FCI~\cite{knowles1984new} method.

\begin{figure*}[t]
    \centering
    \includegraphics[width=0.95\linewidth]{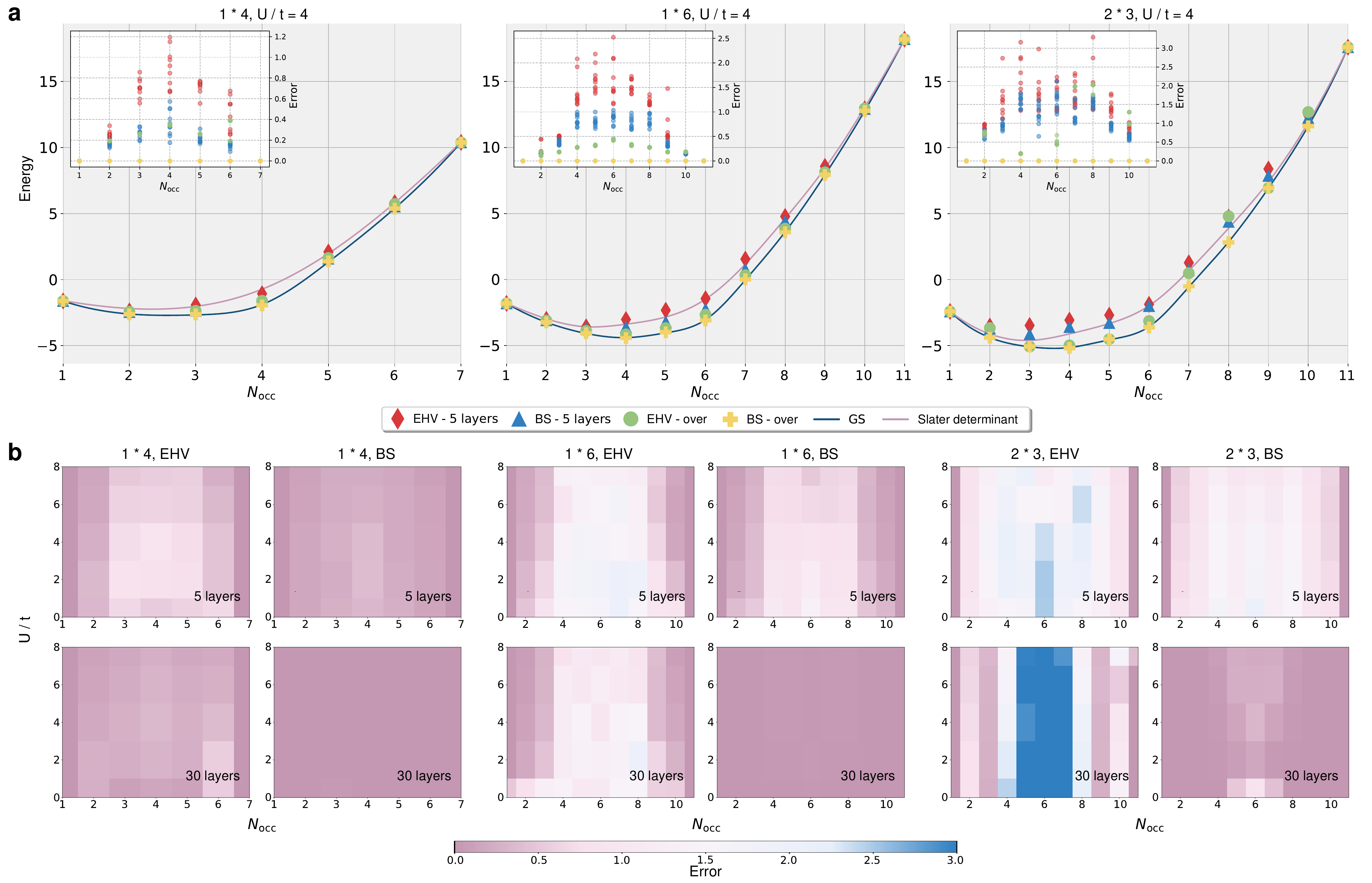}
    \caption{\textbf{Results for the Fermi-Hubbard model.} We conduct experiments for Fermi-Hubbard model instances on lattices of shapes $1\times 4$, $1\times 6$, and $2\times 3$.
    \textbf{a}. Energy of the Fermi-Hubbard model (at $U/t=4$) w.r.t. integer number of occupations. The inset shows the energy error with five random seeds for each method. The "over" in EHV-over and BS-over represents the number of layers beyond the overparameterization bound. \textbf{b}. The energy error of BS and EHV at low circuit depth. In all panels, the X-axis is the occupation number, and the Y-axis is the $U/t$.  All the results are an average of 5 random seeds.}
    \label{fig:Fermi-Hubbard}
\end{figure*}
\subsection{Numerical Results for Fermi-Hubbard Model}
\subsubsection{Problem Definition}
To further illustrate the ansatz-free HWP circuit is capable of solving various Fermionic systems, we conduct experiments on one more iconic model in condensed-matter physics~\cite{leblanc2015solutions}, the Fermi-Hubbard model~\cite{hubbard1963electron}. It is the simplest system that includes non-trivial correlations and has been widely proposed as an early target for variational quantum algorithms~\cite{wecker2015solving,jiang2018quantum,cade2020strategies,cai2020resource,anselme2022simulating,stanisic2022observing}. The Hamiltonian of the Fermi-Hubbard model can be written in the following form:
\begin{equation}\label{eq:Fermi-hubbard}
    H_{HF} = -t\sum_{i,j}\sum_\sigma t_{ij}(a^\dagger_{i\sigma}a_{j\sigma}+a^\dagger_{j\sigma}a_{i\sigma})+U\sum_ia^\dagger_{i\uparrow}a_{i\uparrow}a^\dagger_{i\downarrow}a_{i\downarrow},
\end{equation}
where $i,j$ denotes adjacent sites, and $\sigma\in\{\uparrow,\downarrow\}$ denotes the spin. Each site of the Fermi-Hubbard model contains two fermionic modes, and each mode is mapped to one qubit under the Jordan-Wigner transformation.  Unlike molecular Hamiltonians, the Fermi-Hubbard Hamiltonian is more complex as it includes both intra-site interactions and inter-spin sector couplings.
The first term in Equation~\ref{eq:Fermi-hubbard} is the hopping term, also known as the XY-interaction~\cite{bacon2001encoded,terhal2002classical}
\begin{equation}
    \mathcal{H}_{XY} = a^\dagger_ia_j+a_j^\dagger a_i = \begin{pmatrix}\begin{smallmatrix}
        0 & 0 & 0 & 0\\
        0 & 0 & 1 & 0\\
        0 & 1 & 0 & 0\\
        0 & 0 & 0 & 0
    \end{smallmatrix}\end{pmatrix} = \frac{1}{2}(\sigma_x\otimes\sigma_x+\sigma_y\otimes\sigma_y),
\end{equation}
which is similar to the GR representing the creation and annihilation of the electrons. The second term in Equation~\ref{eq:Fermi-hubbard} is the onsite term, adding a phase to the state $\ket{11}\bra{11}$:
\begin{equation}
    \mathcal{H}_{onsite} = \begin{pmatrix}\begin{smallmatrix}
        0 & 0 & 0 & 0\\
        0 & 0 & 0 & 0\\
        0 & 0 & 0 & 0\\
        0 & 0 & 0 & 1
    \end{smallmatrix}\end{pmatrix},
\end{equation}
which represents the repulsive Coulomb interaction when two electrons with opposite spins occupy the same site. 

Numerical results for Fermi-Hubbard model instances on lattices of shape $1\times 4$, $1\times 6$, and $2\times 3$ are presented in Figure~\ref{fig:Fermi-Hubbard}. Although the $1 \times 6$ and $2 \times 3$ configurations have the same number of sites, the $2 \times 3$ model allows for additional hopping interactions (both vertical and horizontal), making it fundamentally different from the 1D configuration. To assess the efficiency of the proposed HWP ansatz, we compared its performance against the ground state (GS, the exact energy), the Slater determinant state (classical ansatz)~\cite{slater1929theory,helgaker2013molecular}, and the Efficient Hamiltonian Variational (EHV) ansatz\cite{stanisic2022observing}, a VQE approach utilizing both terms in the problem Hamiltonian by applying NN-connected XY-interaction and onsite layers alternately to achieve the desired accuracy. 

\subsubsection{Results}
We first computed the energy for all instances across different occupation numbers at $U/t=4$, focusing on scenarios with both extremely low circuit depths (5 layers) and sufficiently deep circuits to achieve overparameterization (100 layers for 4 sites and 200 layers for 6 sites). Maintaining a relatively low circuit depth is crucial for exploring the practical utility. The results in Figure~\ref{fig:Fermi-Hubbard}a reveal that the BS ansatz consistently outperforms other baseline methods across all circuit depths, with errors below $1\times 10^{-10}$ in the overparameterized regime. We further examined the robustness of the BS ansatz by varying $U/t$ in Figure~\ref{fig:Fermi-Hubbard}b. We found that increasing the number of layers steadily improves its performance, whereas the EHV ansatz shows limited improvement. Additionally, it appears that the noninteracting ($U/t=0$) Fermi-Hubbard model poses more challenges for both VQE methods when underparameterized, likely due to the introduction of unnecessary phases. 
% Nonetheless, we demonstrate that the BS ansatz consistently outperforms existing Hamiltonian variational ansätze, delivering superior results even with a minimal number of layers.

The results in Figure~\ref{fig:Fermi-Hubbard} reveal a profound implication of our universality theorem. Despite the distinct physical origins of the Fermi-Hubbard and molecular Hamiltonians, they share the same underlying symmetry constraints (particle number conservation). Crucially, the efficacy of the BS ansatz is invariant to the specific form of the transition operators in the Hamiltonian. Whether the system is governed by molecular integrals or lattice hopping parameters, the solution resides within the same mathematical subspace.

As observed in Figure~\ref{fig:Fermi-Hubbard}a, while the EHV ansatz struggles to eliminate residual errors due to its reliance on a fixed Trotter-like structure, the BS ansatz achieves near-exact solutions ($\sim 10^{-10}$ error) in the overparameterized regime. The heatmaps in Figure~\ref{fig:Fermi-Hubbard}b further corroborate this robustness, showing that the BS ansatz maintains superior accuracy across the entire phase diagram, from metallic to Mott-insulating regimes (varying $U/t$). This reinforces the core insight that by ensuring complete controllability over the target subspace, the BS ansatz serves as a universal, problem-independent solver. It bypasses the need to tailor the ansatz structure to the specific interaction terms of the Hamiltonian, offering a unified and rigorous solution for diverse fermionic simulation tasks.

While our theoretical derivations rigorously establish that absolute subspace universality (i.e., zero approximation error) necessitates a circuit depth scaling with $d_k^2$, it is crucial to recognize that this $\mathcal{O}(d_k^2)$ bound represents the extreme worst-case scenario. For highly complex systems at half-filling, $d_k$ remains exponentially large, seemingly implying an intractable depth requirement. However, in practical machine learning and chemistry applications, achieving the absolute zero-error exact state is rarely strictly required. Instead, crossing a specific precision threshold (e.g., chemical accuracy, $1 \times 10^{-3}$ Ha) is practical for obtaining a chemically meaningful result. Our empirical evaluations, particularly the ultra-shallow (5 layers) Fermi-Hubbard experiments in Figure 8b, provide a profound insight into this practical regime. Even when operating far below the theoretical overparameterization threshold, the proposed BS ansatz demonstrates exceptional representational efficiency. By natively constraining the optimization exclusively to the physically valid subspace, the BS architecture effectively compresses the necessary solution manifold. Consequently, it consistently outperforms Hamiltonian-driven methods like EHV at identical, highly restricted circuit depths. This establishes a highly practical paradigm that while the BS ansatz offers a theoretical guarantee of exactness at exponential depth, it simultaneously delivers a superior, highly hardware-efficient heuristic approximation at extreme shallow depths.

\section{Conclusion}
This work rigorously establishes the necessary and sufficient conditions for an HWP ansatz to achieve subspace universality, along with a comprehensive analysis of its trainability. These conditions can be further extended to accommodate different physical qubit connectivity configurations and tailored for specific values of $n$ and $k$, enabling the design of HWP operators with simpler decompositions for broader applicability in real-world problems.
The proposed ansatz represents a significant step toward addressing the challenges of VQEs, striking a critical balance between expressivity and trainability while offering inherent error detection capabilities. Any deviation in Hamming weight serves as a direct indicator of bit-flip errors, making it highly suitable for noisy environments.
By providing a mathematically interpretable, noise-resilient framework, this ansatz presents a compelling pathway for advancing discussions of VQE supremacy in the near term. 

\section*{Acknowledgement}
The work was partly supported by NSFC (72342023) and CPS-Yangtze Delta Region Industrial Innovation Center of Quantum and Information Technology-MindSpore Quantum Open Fund.

\clearpage
\scriptsize
\bibliography{ref}
\bibliographystyle{ieeetr}

\begin{IEEEbiography}
[{\includegraphics[width=1in,height=1.25in,clip,keepaspectratio]{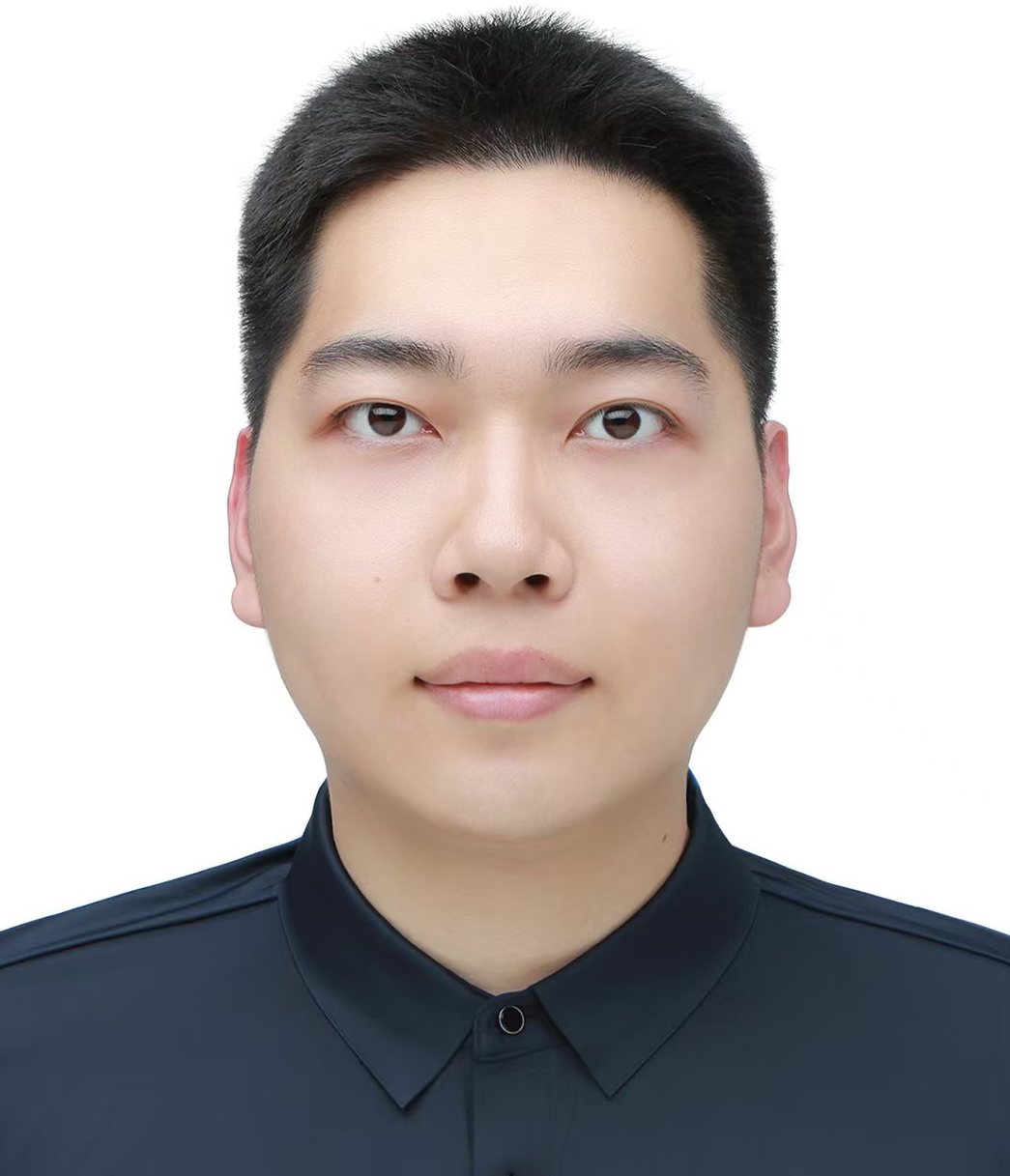}}] 
{Ge Yan} (M'25) received the PhD degree and B.S. (ACM Honored Program) degree both in Computer Science from Shanghai Jiao Tong University, Shanghai, China in 2025 and 2020, respectively. He is currently a postdoctoral researcher working in the lab led by Dr. Yuxuan Du and Professor Dacheng Tao at Nanyang Technological University, Singapore. His research interests are quantum AI and AI for quantum, with first-author papers in ICML, ICLR, NeurIPS, AAAI, SIGKDD, etc.
\end{IEEEbiography}
\vspace{-30pt}
\begin{IEEEbiography}
[{\includegraphics[width=1in,height=1.25in,clip,keepaspectratio]{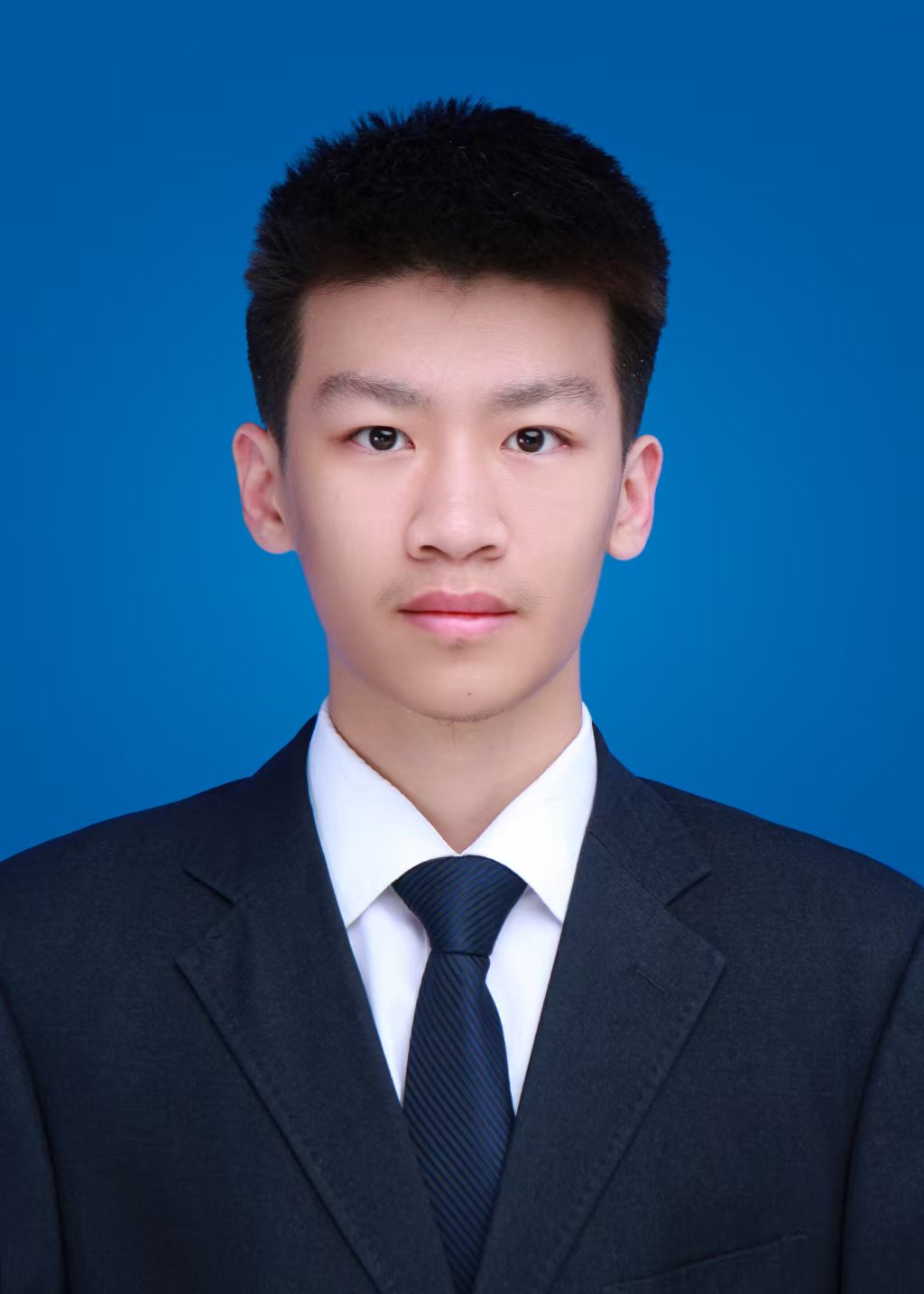}}]
{Kaisen Pan} received B.S. in Computer Science (ACM Honored Program) from Shanghai Jiao Tong University in 2024. He is currently an independent researcher. His research interests include quantum computing and quantum artificial intelligence. During his undergraduate studies, he contributed to 4 papers in top venues. He received a gold medal in the ICPC Regional Contest and a gold medal in the East Asia Final Contest.
\end{IEEEbiography}
\vspace{-30pt}
\begin{IEEEbiography}
[{\includegraphics[width=1in,height=1.25in,clip,keepaspectratio]{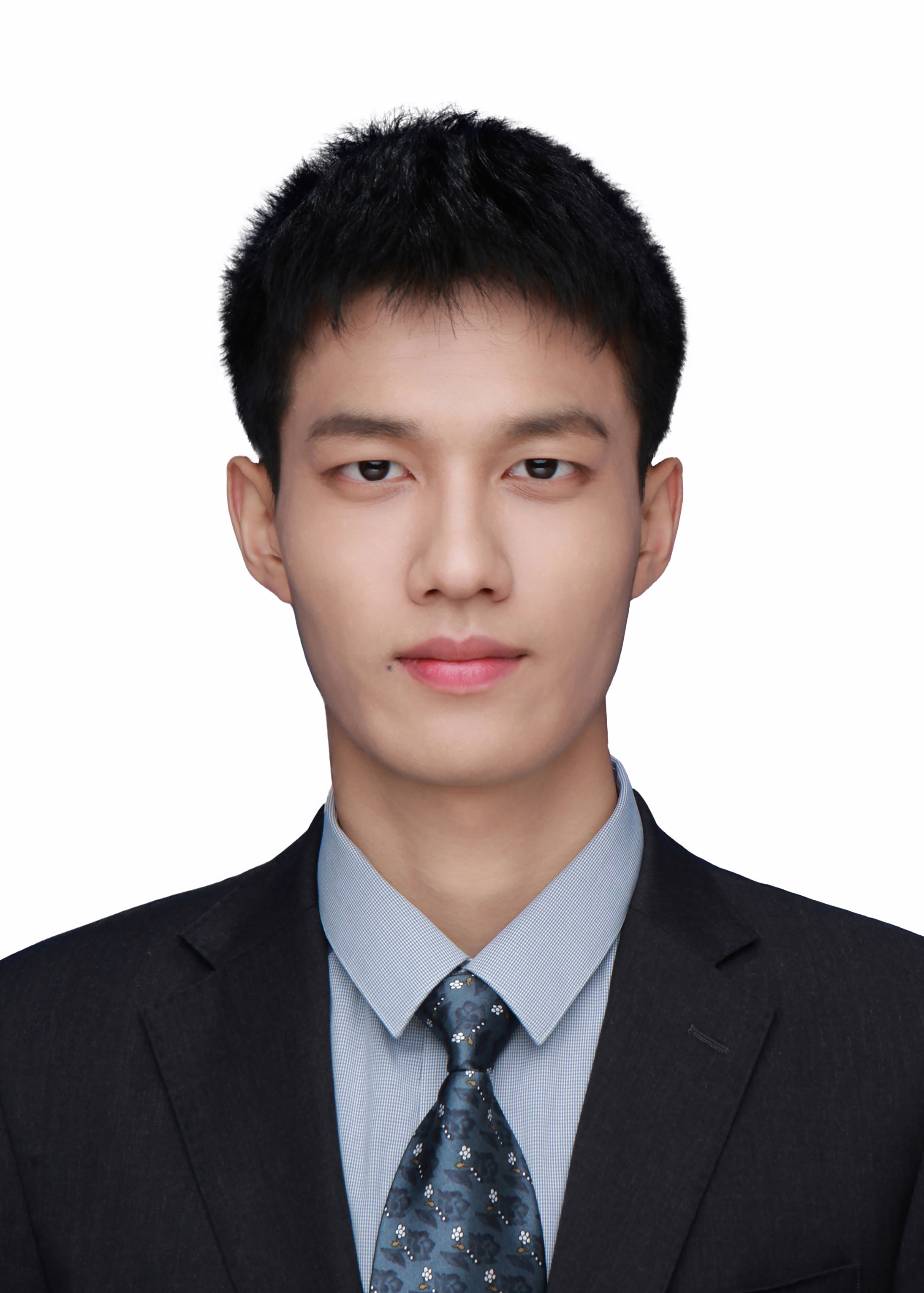}}]
{Ruocheng Wang} is a PhD student with the School of AI, Shanghai Jiao Tong University, Shanghai, China, and Shanghai Innovation Institute (SII). Before that, he received his BSc degree in Physics from the Zhiyuan College, Shanghai Jiao Tong University. His research interests include quantum algorithms, machine learning, and PDE solving. He has published several papers in venues such as ICML and NeurIPS. He was awarded the National Scholarship and Shanghai Outstanding Graduate.
\end{IEEEbiography}
\vspace{-30pt}
\begin{IEEEbiography}
[{\includegraphics[width=1in,height=1.25in,clip,keepaspectratio]{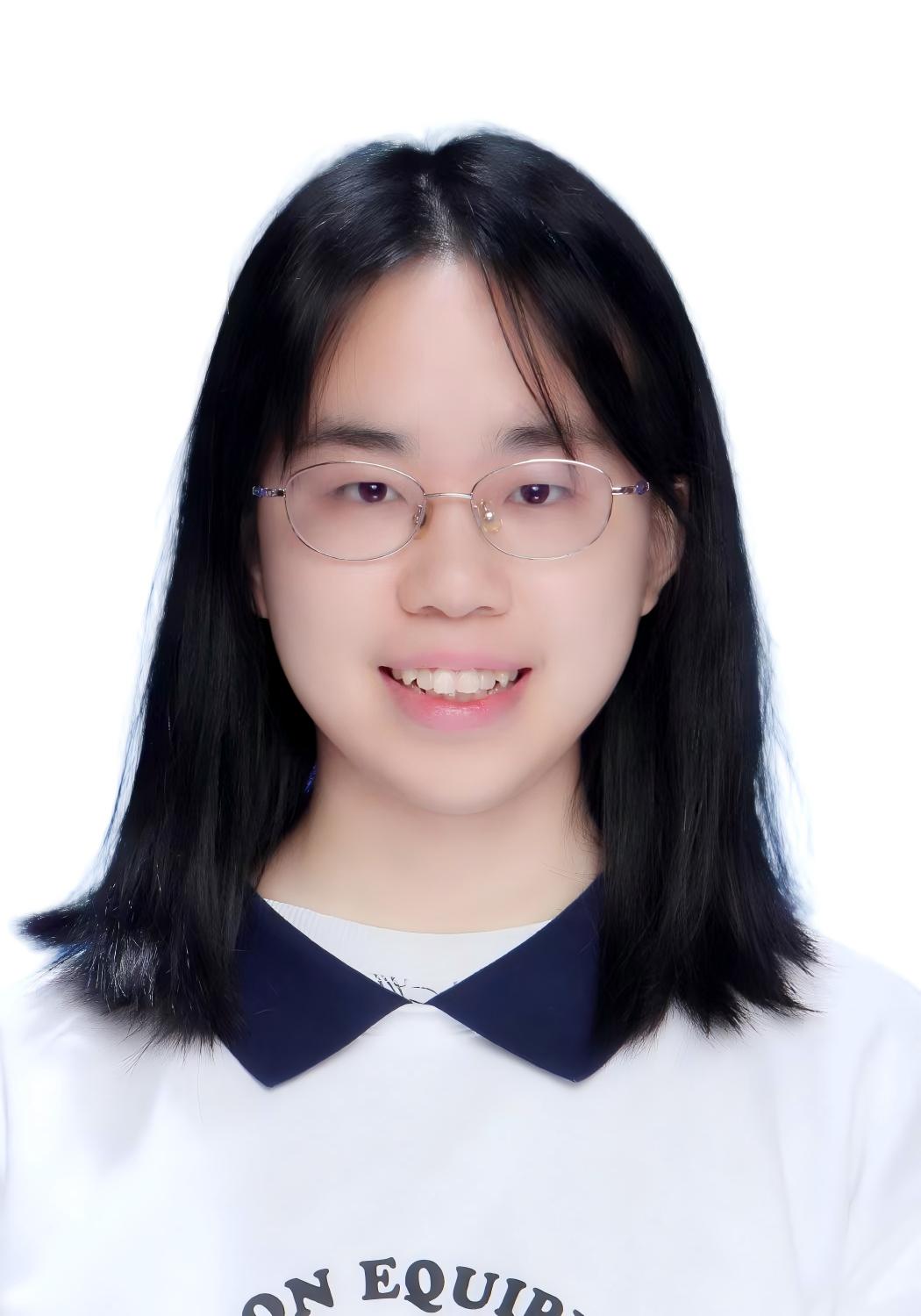}}]
{Mengfei Ran} is a PhD student with the School of Artificial Intelligence at Shanghai Jiao Tong University. Before that, she received her BSc degree in Physics from the School of Physics and Astronomy, Shanghai Jiao Tong University. Her research interests include quantum circuit compilation, quantum machine learning, and quantum algorithms for scientific computing. She has published in Top venues such as NeurIPS.
\end{IEEEbiography}
\vspace{-30pt}
\begin{IEEEbiography}
[{\includegraphics[width=1in,height=1.25in,clip,keepaspectratio]{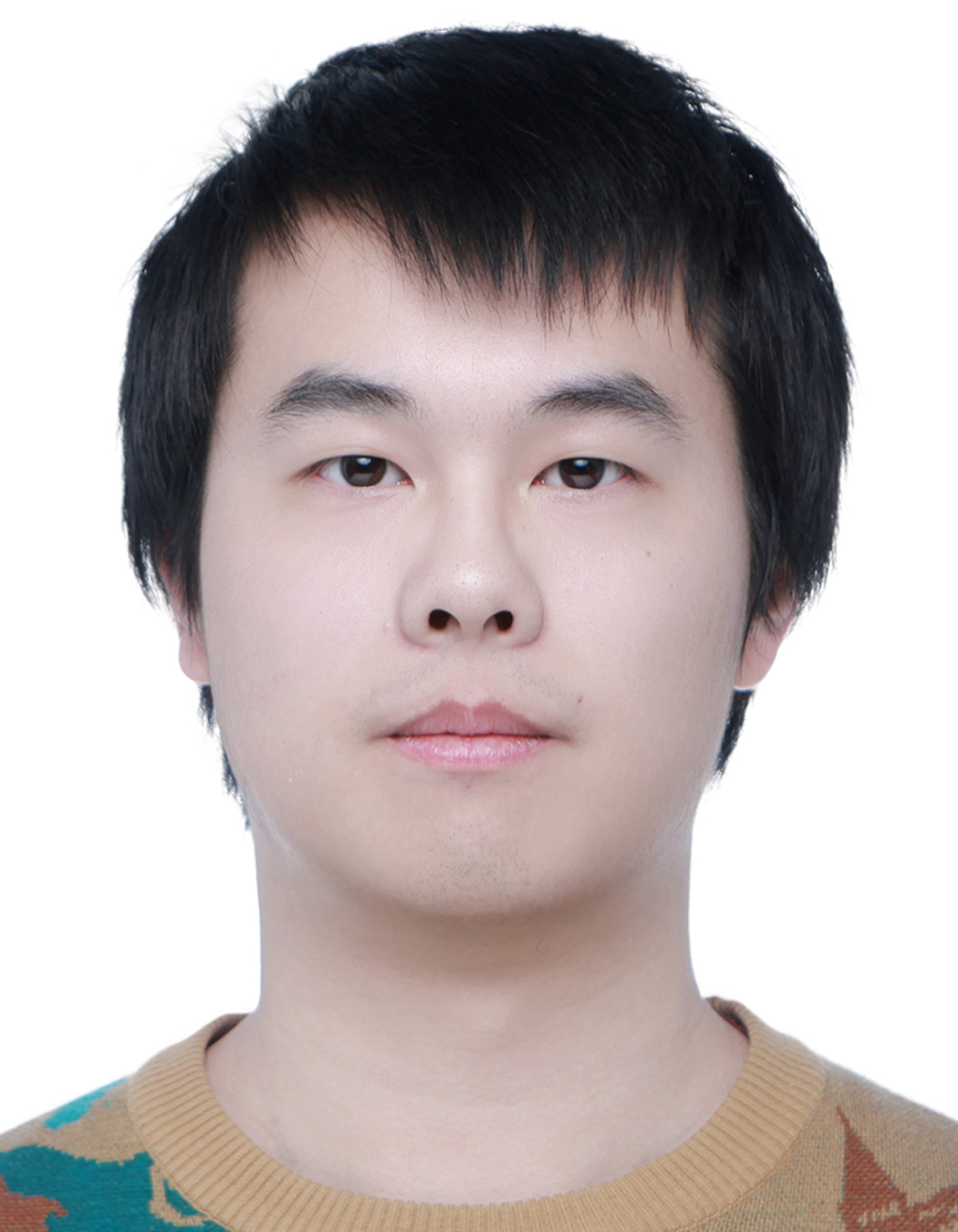}}]
{Hongxu Chen} received his B.S. in Information Engineering from Shanghai Jiao Tong University, China, in 2022, and is pursuing a Master's degree in Data and Artificial Intelligence at Institut Polytechnique de Paris, France. Previously, he was a Research Assistant at ReThinkLab, Shanghai Jiao Tong University, working on variational quantum algorithms. He completed an internship at Laboratoire d’informatique de l’École polytechnique, Institut Polytechnique de Paris, focusing on ZX-calculus. He has published in top venues such as ICLR.
\end{IEEEbiography}
\vspace{-30pt}
\begin{IEEEbiography}
[{\includegraphics[width=1in,height=1.25in,clip,keepaspectratio]{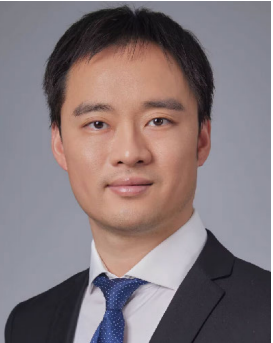}}] 
{Junchi Yan} (S'10-M'11-SM'21) is a Professor and Associate Director with School of AI, Shanghai Jiao Tong University, Shanghai, China. Before that, he was a Research Staff Member with IBM and later an affiliated consultant Researcher with AWS AI Lab. His research interests are machine learning. He is the Associate Editor for IEEE TPAMI/TNNLS/TEVC, JMLR and Pattern Recognition. He received IEEE CS AI'10 to Watch, IEEE CIS Outstanding Early Career Award, CVPR/IROS Best Paper Candidate, ACL Outstanding Paper. He is a Fellow of IAPR, and on the board of ICML and the Program Co-Chair of ACM Multimedia 2026. He has published 20+ papers in top venues in quantum AI and AI for quantum and is the featued cover author of IEEE Xplorer.
\end{IEEEbiography}

\end{document}